\documentclass[12pt]{article}
\usepackage{graphicx} 
\usepackage{algorithm2e}
\usepackage{float}
\usepackage{booktabs}
\usepackage{multirow}
\usepackage{tikz}
\usepackage{amsthm}
\usepackage{amssymb}
\usepackage{threeparttable}

\usepackage{amsmath}
\usepackage{caption}
\usepackage{bm}
\usepackage{subcaption}
\usepackage{amsmath}
\RequirePackage[colorlinks,citecolor=blue,linkcolor=blue,urlcolor=blue,pagebackref]{hyperref}
\usepackage{natbib}

\newtheorem{proposition}{Proposition}
\newtheorem{theorem}{Theorem}

\newtheorem{definition}{Definition}

\def\spacingset#1{\renewcommand{\baselinestretch}
{#1}\small\normalsize} \spacingset{1.5}

\author{Sam van Meer$^{\text{a}}$, Nick Koning$^{\text{a}}$ \\ \footnotesize {\vspace{0pt} $^{\text{a}}$Econometric Institute, Erasmus University Rotterdam; Tinbergen Institute, Netherlands.}\\
} \vspace{12pt}\date{April 2025}

\begin{document}

\title{Real-time Program Evaluation using Anytime-valid Rank Tests*}

\maketitle

\begingroup
\renewcommand\thefootnote{*}
\footnotetext{Corresponding author: Sam van Meer, \texttt{vanmeer@ese.eur.nl}. We thank Alberto Abadie, Dick van Dijk, Stan Koobs, Bernhard van der Sluis, and Wendun Wang for their valuable feedback. We are also grateful to participants at IAAE 2024, EAYE 2024, Camp Econometrics 2024, and the ISEO Summer School 2024 for their helpful comments.}

\endgroup

\begin{abstract}
    \noindent Counterfactual mean estimators such as difference-in-differences and synthetic control have grown into workhorse tools for program evaluation.
    Inference for these estimators is well-developed in settings where all post-treatment data is available at the time of analysis.
    However, in settings where data arrives \emph{sequentially}, these tests do not permit real-time inference, as they require a pre-specified sample size $T$.
    We introduce real-time inference for program evaluation through anytime-valid rank tests. 
    Our methodology relies on interpreting the absence of a treatment effect as exchangeability of the treatment estimates.
    We then convert these treatment estimates into sequential ranks, and construct optimal finite-sample valid sequential tests for exchangeability.
    We illustrate our methods in the context of difference-in-differences and synthetic control. In simulations, they control size even under mild exchangeability violations.
    While our methods suffer slight power loss at~$T$, they allow for early rejection (before $T$) and preserve the ability to reject later (after $T$).
\end{abstract}

\newpage 
\tableofcontents

\section{Introduction}\label{s1}
    We sequentially observe outcomes $Y^T = (Y_1, Y_2, \dots, Y_{T_0}, \dots, Y_{T})$, of some individual or \emph{unit} that undergoes a treatment at time $T_0$, up until some possibly infinite final time $T$.
    Starting from the treatment time $T_0$ onwards, we are interested in testing the hypothesis that the treatment has no effect on our outcome $Y^T$.
    
    We monitor this by comparing the post-treatment outcomes to the pre-treatment outcomes.
    To make these comparable, we assume they are appropriately transformed to be exchangeable in absence of a treatment effect.
    Depending on the context, this means the treatment outcomes may be directly observed, or they may be estimates obtained from some counterfactual mean method such as difference-in-differences or synthetic control, as we demonstrate later.
    We now formulate the hypothesis of no treatment effect $H_0$ as
    \begin{align}
        H_0 : Y_1, Y_2, \dots, Y_{T_0}, \dots, Y_{T} \text{ are exchangeable}.
    \end{align}
    Here, exchangeability means that for any finite $t\leq T$, the distribution of data $Y_1, \dots Y_{t}$ is invariant under re-ordering of its observations.
    The advantage of this abstract formulation of the notion of `no treatment effect' is that it can not only be used to capture changes in the mean, but also other changes in the distribution.

    This problem is well-studied, but only for situations where the number of observations $T$ is pre-specified or, equivalently, independently specified.
    For example, \citet{chernozhukov2021exact} and \cite{abadie2021synthetic} propose a permutation $p$-value $p_T$ that is valid for every pre-specified $T > T_0$:
    \begin{align}\label{eq:type-I_error}
        \sup_{T > T_0} \mathbb{P}^{H_0}(p_T \leq \alpha) \leq \alpha, \text{ for all $\alpha > 0$}.
    \end{align}
    This guarantees that at every point in time $T$, the probability that the $p$-value $p_T$ is below $\alpha$ is at most $\alpha$.

    When evaluating the impact of a treatment, the requirement to pre-specify the number of observations $T$ is restrictive.
    Indeed, if we set $T$ too large, then we incur the cost of collecting too much data, and an opportunity cost due to the delay.
    On the other hand, setting $T$ too small may lead to an underpowered test, which comes with the opportunity cost associated to not discovering an effective treatment.

    To prevent collecting too much data, it may be tempting to reject the hypothesis $H_0$ at the first time $T$ at which $p_T \leq \alpha$.
    However, for traditional tests, such real-time monitoring is prohibited, as this makes the number of observations $T$ data-dependent while the Type-I error guarantee in \eqref{eq:type-I_error} only holds for data-independent $T$.
    The probability that a standard $p$-value \emph{ever} dips below $\alpha$ is typically much larger than $\alpha$.

    \subsection{Contribution}
        The primary contribution of this paper is to develop \emph{anytime-valid} $p$-values $\mathfrak{p}_T$ for real-time program evaluation.
        In contrast to a regular $p$-value as in \eqref{eq:type-I_error}, such an anytime-valid $p$-value is valid at \emph{arbitrary data-dependent times $T$}.
        In particular, the probability under $H_0$ that $\mathfrak{p}_T$ \emph{ever} dips below $\alpha$ is bounded by $\alpha$:
        \begin{align} \label{eq:anytimeval}
            \mathbb{P}^{H_0}\left(\inf_{T > T_0} \mathfrak{p}_T \leq \alpha\right) \leq \alpha.
        \end{align}
        A consequence is that we need not pre-specify a number of observations, but we can test at every moment in time without invalidating the Type I error guarantee.
    
        We tailor our anytime-valid $p$-values to the key feature of this problem: the fact that at the moment we start testing, we already have an entire batch of $T_0$ pre-treatment observations available to us.
        If we order the pre-treatment observations by magnitude, the exchangeability hypothesis $H_0$ implies that the first incoming post-treatment observation is equally likely to land in any position between these pre-treatment observations.
        We exploit this to design two frameworks, centered around two alternative hypotheses, for which we construct our anytime-valid $p$-values. \\
        
        \noindent \textbf{Post-exchangeable alternative} \\
        In the first framework, we construct our $p$-values by keeping track of how the newly incoming post-treatment observations rank among the pre-treatment observations.
        As this framework does not depend on the ranking of the post-treatment data among themselves, this can be translated into an alternative hypothesis in which the post-treatment data are exchangeable:
        \begin{align}
            H_1^{\text{post-exch.}} : Y_1, Y_2, \dots, Y_{T_0}, \dots, Y_T \textnormal{ are not exchangeable, but } Y_{T_0 + 1}, \dots, Y_T \textnormal{ are}.
        \end{align}
        An interpretation of this alternative is that after time $T_0$, the treatment shifts the data generating process to some new stable equilibrium from which the post-treatment data are drawn. \\
    
        \noindent \textbf{Generic alternative} \\
        In our second framework, we do not just consider how the post-treatment data ranks among the pre-treatment data, but also how the post-treatment data ranks among itself.
        This translates into the `generic' alternative hypothesis that the treatment outcomes are simply not exchangeable:
        \begin{align}
            H_1^{\text{generic}} : Y_1, Y_2, \dots, Y_{T_0}, \dots, Y_T \textnormal{ are not exchangeable}.
        \end{align}
        This lends itself better, for example, to situations in which the underlying treatment effect fluctuates or continues to grow over time. Note that this hypothesis nests the post-exchangeable alternative defined before. Therefore, tests against this generic alternative are generally expected to be less powerful against the narrower post-exchangeable alternative.\\
    
        For both these alternatives, we derive a flexible way to construct an anytime-valid $p$-value.
        In particular, our construction works for \emph{any} choice of non-negative test statistic, that may arbitrarily depend on the previously observed ranks or on external information.
        If the data generating process under the alternative is known, we show how to exploit this to derive a $p$-value with a type of optimal shrinkage rate \citep{grunwald2024safe}.
        As an example, we derive the optimal anytime-valid procedure under a Gaussian alternative.
        If the data generating process under the alternative is not known, we derive an adaptive test statistic that learns the alternative during the sequential testing procedure \citep{fedorova2012plug}.
    
        In addition, we demonstrate how our methodology can be applied in an interactive fixed-effects model.
        There, both difference-in-differences and synthetic control produce exchangeable treatment estimates in absence of a treatment effect under standard assumptions \citep{abadie2021synthetic, chernozhukov2021exact}.
        We use this setting to illustrate our methods in a simulation study, and demonstrate that our methods control size over time when the treatment estimators are (approximately) exchangeable under the null.
        If exchangeability is strongly violated due to presence of serial dependence, we show that size distortions can be mitigated by imposing a block structure, as proposed in the fixed-$T$ setting by \cite{chernozhukov2021exact}.
        In terms of power, we find that the cost of performing anytime-valid inference is limited, especially when the expected effect size is accurately specified. 
        We highlight the benefit of our real-time anytime-valid inference in a discounted utility model, where we find our anytime-valid tests to be preferred over the fixed-$T$ test for a wide range of intertemporal discount factors.

    \subsection{Example: exchangeable treatment estimators}
        In this section, we briefly illustrate how counterfactual and synthetic control (CSC) methods \citep{chernozhukov2021exact} can give rise to exchangeable treatment estimates, so that our methods are applicable.

        We consider a panel with $J + 1$ units over time, where only the first unit $i = 1$ is treated after time $T_0$, and the remaining $J$ units are never treated.
        The observed outcomes of interest $Y^{\text{obs}}_{it}$ are modelled through the potential outcomes framework:
        \begin{equation}
             Y^{\text{obs}}_{it}
            :=
            \begin{cases}
                Y_{it}(1) \quad \text{if $i=1$ and $t>T_0$}, \\
                Y_{it}(0)  \quad \text{else},
            \end{cases}
        \end{equation}
        where $Y_{it}(1)$ is the outcome at time $t$ of unit $i$ if it is treated, and $Y_{it}(0)$ if not.
        Assuming the data is real-valued, we are interested in the causal effect of treatment on the treated unit, which is defined as $\tau_t=Y_{1t}(1)-Y_{1t}(0)$.
        
        Unfortunately, by construction, we only observe $Y_{1t}(0)$ up until time $T_0$ and only observe $Y_{1t}(1)$ beyond time $T_0$: we never observe both $Y_{1t}(1)$ and $Y_{1t}(0)$ at any particular time $t$.
        As a consequence, $\tau_{t}$ cannot simply be computed, but needs to be inferred. 
        CSC methods estimate $\tau_t$ by constructing an approximation $\widehat{Y}_{1t}(0)$ of the counterfactual outcome $Y_{1t}(0)$ of the treated unit, had it not been treated.
        This approximation is then used to construct the treatment estimator $\widehat{\tau}_t = Y_{1t}^{\text{obs}} - \widehat{Y}_{1t}(0)$. 

        To construct the counterfactual approximation, it is common to use the observed outcomes $Y^{\text{obs}}_{it}$ of the untreated units. 
        The Synthetic Control Method \citep{abadie2003economic}, then constructs the approximation of the counterfactual $\widehat{Y}_{1t}(0)$ as a linear combination of the untreated outcomes, $\widehat{Y}_{1t}(0)=\sum_{i=2}^{J+1}w_i Y^{\text{obs}}_{it}$ for some appropriate weights $w_i$.  

        Let $u_t$ denote the approximation error of this proxy, such that $\widehat{\tau}_t=\tau_t+u_t$. 
        If $u_1, u_2, ...$ are exchangeable under $H_0$, then we can recast the null-hypothesis $H_0: \tau_{T_0+1}=\tau_{T_0+2}=...=0$ into $H_0: \{\widehat{\tau}_t\}_{t>0}$ is exchangeable \citep{chernozhukov2021exact}. 
        Hence, sequential testing for a treatment effect can be done by sequentially testing for exchangeability of the CSC estimator. 
        Later, we show how the interactive fixed effects model by \cite{bai2009panel} can give rise to these exchangeable approximation errors.       

    \subsection{Related literature}
        As mentioned, our work is related to that of \cite{chernozhukov2021exact}, which studies inference for counterfactual mean models.
        Their work is also based on exchangeability over time, and relies on the same assumptions.
        However, they only consider testing at a fixed number of observations $T$, whereas our methods permit real-time inference with an adaptive number of observations.
    
        Our paper also relates to the new \emph{synthetic design} literature, which applies the synthetic control method in experimental setting \citep{abadie2021synthetic,doudchenko2019designing,doudchenko2021synthetic}.
        Compared to traditional A/B testing, synthetic control has the advantage that it can account for regional market effects and spillovers.
        A key focus of this literature is on experimental design, particularly in optimizing the selection of treatment units.
        Since our methods allow for endogenous treatment selection, they can be integrated with such experimental designs to enable sequential inference in the context of synthetic design.
        
        There already exists some work on sequential causal inference under more restrictive assumptions. 
        \cite{maharaj2023anytime} study anytime-valid inference for A/B testing, but they assume a multiple-treated-units framework in which the outcome of a new treated unit is observed with every observation. 
        In contrast, our methodology also works for a single treated unit that undergoes one treatment. 
        Similarly, \cite{ham2022design} develop \textit{asymptotically} anytime-valid confidence intervals for causal inference. 
        Their methods work in the time-series/panel setting, but require that treatment is randomly assigned at every time point. 
        Our method is different from theirs in that we allow for endogenous treatment selection, and do not require new treatment assignment at every time point. 
        Also, our proposed methods do not rely on asymptotic approximations but have finite sample size guarantees.

        Rank-based tests for sequential exchangeability in an abstract context can be traced back to at least \cite{vovk2003testing}.
        However, to the best of our knowledge, they have not been previously applied in treatment evaluation and have generally been surprisingly underappreciated in statistical inference.
        Other recent applications include testing independence \citep{henzi2024rank} and outlier detection \citep{bates2023testing}.
        Moreover, \citet{koning2024post} studies how this can be generalized beyond testing exchangeability to other forms of invariance.

        Compared to existing work on rank-based anytime-valid testing of exchangeability, our work has two unique features.
        First, we have an entire batch of pre-treatment observations available to us when we start testing.
        Second, we exploit the pre-/post-treatment structure in an alternative hypothesis under which we only consider the ranks of the post-treatment data among these pre-treatment observations.
        Interestingly, \citet{fischer2024sequential} study a related approach in a completely unrelated context: reducing the number of Monte Carlo simulations in a Monte Carlo-based hypothesis test. 
        Abstractly speaking, their approach can be interpreted as the special case of our method, in which only a single pre-treatment period is considered.
        
    \section{Background: \textit{e}-values and anytime-valid inference}\label{sec:background}
        To simplify the discussion, we consider testing a simple null hypothesis that contains a single distribution $\mathbb{P}$ throughout this section.
        While our null hypothesis of exchangeability is highly composite, we show that it can be reduced to a simple hypothesis in Section \ref{sec:seqtest}.
        
        The construction of our anytime-valid sequential tests relies on a recently introduced measure of evidence called the $e$-value \citep{shafer2011test, grunwald2024safe,howard2021time,vovk2021values, ramdas2023game,grunwald2024safe,koning2024continuous}.
        An $e$-value is typically defined as a non-negative random variable $e$ with expectation bounded by 1 under the null hypothesis:
        \begin{align*}
            \mathbb{E}[e] \leq 1.
        \end{align*}
        
        The definition of the $e$-value does not immediately convey its usefulness in a testing context, but a simple application of Markov's inequality shows that we can construct a valid test from an $e$-value.
        In particular, a test that rejects when $e \geq 1/\alpha$ is valid at level $\alpha > 0$:
        \begin{align*}
            \mathbb{P}(e \geq 1/\alpha)
                \leq \alpha \mathbb{E}[e]
                \leq 1.
        \end{align*}
        
        Moreover, $e$-values inherit simple merging rules from the properties of the expectation: the average of two $e$-values and the product of independent $e$-values is also an $e$-value. 
        A generalization of this product-merging property combined with the possibility to convert an $e$-value into a test forms the basis for a powerful machinery for constructing anytime-valid sequential tests \citep{ramdas2020admissible}.
        We explain the key parts of this machinery in this section.
    
    \subsection{Sequential \textit{e}-values and test martingales}
        In the context of sequential testing, we must carefully describe the available information at every moment in time.
        This means that our sample space is not just equipped with some sigma algebra of events, but with an entire filtration $\{\mathcal{F}_t\}_{t \geq 0}$ that describes the available information $\mathcal{F}_t$ at time $t \geq 0$.
        With such a filtration, we can speak of a sequential $e$-value $e_t$ at time $t$, if it is an $e$-value with respect to the available information $\mathcal{F}_{t-1}$.
        Without loss of generality, we set $e_0 = 1$ out of convenience.
    
        \begin{definition}[Sequential $e$-value]
            For $t > 0$, a sequential $e$-value $e_t$ with respect to the filtration $\mathcal{F} = \{ \mathcal{F}_t \}_{t \geq 0}$, is a non-negative random variable such that under the null hypothesis:
            \begin{equation}
                \mathbb{E} \left[e_t | \mathcal{F}_{t-1}\right] \leq 1.
            \end{equation}
        \end{definition} 
    
        In order to construct an anytime-valid sequential test out of these sequential $e$-values, we start by stringing them together through multiplication.
        In particular, the running product $W_t = \prod_{s = 0}^t e_s$ of a sequence of sequential $e$-values $e_0, e_1\dots$, where $e_0=1$, constitutes a non-negative supermartingale that starts at 1.
        Such martingales are also referred to as \emph{test martingales}, for their applications in testing \citep{shafer2011test}.
    
        \begin{definition}
            Let $\{W_t\}_{t \geq 0}$ be a non-negative random process adapted to a filtration $\mathcal{F}=\{\mathcal{F}_t\}_{t\geq 0}$.
            We say that this is a test martingale for our null hypothesis if it starts at 1, is integrable at every $t$, and
           \begin{align*}
                \mathbb{E}\left[W_t | \mathcal{F}_{t-1}\right] \leq W_{t-1}, \quad \text{for all $t > 0$}.
           \end{align*}
        \end{definition}
    

    \subsection{From test martingale to anytime-validity through Ville's inequality}
        Recall that our aim is to construct an anytime-valid $p$-value $\mathfrak{p}_t$, as defined in (\ref{eq:anytimeval}).
        It turns out that this can be accomplished by simply setting the $p$-value equal to the reciprocal of a test martingale: $\mathfrak{p}_t:=\frac{1}{W_t}$.
        Indeed, this follows from an application of Ville's inequality:
        \begin{align}\label{ineq:ville}
            \mathbb{P}\left(\inf_{t\geq0}\mathfrak{p}_t\leq\alpha\right)
                = \mathbb{P}\left(\sup_{t\geq0}W_t\geq\frac{1}{\alpha}\right)
                \leq \alpha, \quad \textnormal{for all } \alpha > 0,
        \end{align}      
        under the null hypothesis \citep{ville1939etude,grunwald2024safe}.
        In words, this means that for the $p$-value $\mathfrak{p}_t$, the probability that it \emph{ever} dips below $\alpha$ is bounded by $\alpha$.
        As such $p$-values are actually stochastic processes, they are sometimes also referred to as $p$-processes.

        Ville's inequality, the inequality in \eqref{ineq:ville}, can be viewed as a sequential generalization of Markov's inequality for martingales.
        If we interpret Markov's inequality as stating that the odds to double our money in a fair bet are no more than 50/50, then Ville's inequality states that the odds to \emph{ever} double our money in a sequence of fair bets is no more than 50/50.
    
        A testing interpretation of Ville's inequality is that the probability our test martingale $W_t$ \emph{ever} exceeds the critical value $1/\alpha$ is bounded by $\alpha$.
        This is true in particular for the first data-dependent time $\inf\{t > 0 \mid W_t \geq 1/\alpha\}$ at which the test martingale exceeds the threshold, but also for any other possibly data-dependent time.
        A consequence is that we can continuously evaluate whether our test martingale exceeds $1/\alpha$, which allows for early rejection if the evidence is strong or continuation if the evidence remains weak.

    \subsection{Constructing good anytime-valid \textit{p}-values and test martingales}\label{sec:good_martingales}
        It remains to discuss how to construct a good anytime-valid $p$-value.
        Ideally, we would like a $p$-value that shrinks `as quickly as possible' when the alternative hypothesis is true.
        As our anytime-valid $p$-value is the reciprocal of a test martingale, this is equivalent to finding a test martingale that grows as quickly as possible.

        Before we can even define what it means for a test martingale to grow as quickly as possible, we must first define when it should grow as quickly as possible.
        To do this, let us assume for now that if the alternative hypothesis is true, then the data are sampled from the distribution $\mathbb{Q}$.
        This means our test martingale should grow as quickly as possible if the data comes from $\mathbb{Q}$.

        As our test martingale is the running product of sequential $e$-values, the problem can be decomposed into choosing sequential $e$-values that lead to a high growth rate.
        Let $\mathbb{Q}_t$ denote the conditional distribution given the available information $\mathcal{F}_{t-1}$ at time $t -1$, which is the moment at which we must choose the sequential $e$-value that we will use at time $t$.
        
        A naive strategy would be to choose the sequential $e$-value that has the largest expectation under $\mathbb{Q}_t$.
        Unfortunately, this sequential $e$-value is usually zero with a substantial positive probability.
        This is problematic: if we encounter a single zero along the way, this will multiply the current value of our martingale by zero, which means it will equal zero and can never grow again.
        For this reason, the anytime-validity literature typically considers different targets.
        
        A popular target is to maximize the expected logarithm $\mathbb{E}^{\mathbb{Q}_t}[\log e_t]$, which is equivalent to maximizing the geometric expectation.
        The motivation for this target comes from the i.i.d. setting, where this maximizes the long-run growth rate \citep{kelly1956new}.
        Indeed, if $e_1, \dots e_t$ are i.i.d., then a combination of the strong law of large numbers and the continuous mapping theorem tells us that this maximizes the average asymptotic growth rate:
        \begin{align*}
            (W_t)^{1/t}
                = \prod_{s = 0}^t (e_s)^{1/t}
                = \exp\left\{\ln \prod_{s = 0}^t (e_s)^{1/t}\right\}
                = \exp\left\{\tfrac{1}{t}\sum_{s = 0}^t \log  e_s\right\}
                \rightarrow \exp\left\{\mathbb{E}\log e_1\right\},
        \end{align*}
        
        Remarkably, for a simple null hypothesis with distribution $\mathbb{P}$ and alternative $\mathbb{Q}$, where $\mathbb{P}\gg \mathbb{Q}$, this `log-optimal' sequential $e$-value simply equals the conditional likelihood ratio \citep{kelly1956new, koolen2022log, grunwald2024safe, larsson2024numeraire}
        \begin{align} \label{eq:logopt}
            e_t
                = \frac{d\mathbb{Q}_t}{d\mathbb{P}_t}.
        \end{align}
        This means that the growth-rate maximizing test martingale is the product of these conditional likelihood ratios.

    \subsection{Composite alternatives}
        In our application, the alternative hypothesis $H_1$ is typically composite, as it is often hard to know beforehand how effective a treatment will be.
        This means that we do not know the precise distribution $\mathbb{Q}$ from which the data are sampled when the alternative hypothesis is true.
        Fortunately, we can use the sequential nature of the setting to adaptively learn about this distribution while we are testing, and update our choice of sequential $e$-values accordingly.

        One such approach is the plug-in approach, where we plug our current estimator of $\mathbb{Q}$, given the available information, in for $\mathbb{Q}$ in the target.
        Another approach is the method of mixtures, where our sequential $e$-value at time $t$ is an adaptive mixture over multiple candidate sequential $e$-values at time $t$ \citep{ramdas2023game, ramdas2020admissible}.
        Here, an adaptive mixture means that we can update the mixing weights based on the historical performances of each of these $e$-values.

        In case we consider $k \in \mathbb{N}$ candidate sequential $e$-values, there exists a weighting scheme for these $e$-values with an attractive minimax regret property, as demonstrated in Appendix 
        \ref{app:minimax}.
        This adaptive weighting scheme, when applied to the candidate $e$-values, reduces to taking a simple average over their respective martingales. 
        The regret bound states that the average growth of this mixture martingale at time $t$ is at most $\frac{\log k}{t}$ away from that of the $e$-value candidate with highest average growth.

\section{Rank-based anytime-valid tests for program evaluation} \label{sec:seqtest}
    In this section, we present our main contribution: the construction of anytime-valid $p$-values for program evaluation.
    In Section \ref{sec:background}, we showed that such an anytime-valid $p$-value $\mathfrak{p}_t$ can be constructed as the reciprocal $\mathfrak{p}_t = 1/W_t$ of a test martingale $W_t$.
    In turn, a test martingale is formed as the running product $W_t = \prod_{s = 1}^t e_s$ of sequential $e$-values $e_1, \dots, e_t$.
    Therefore, we focus in this section on the construction of sequential $e$-values for program evaluation, as they are the building blocks for anytime-valid $p$-values. 

    Recall that we sequentially observe the treatment estimates $Y_1, Y_2, \dots, Y_{T_0}, \dots$, where $Y_1, \dots Y_{T_0}$ are considered untreated, and $Y_{T_0 + 1},  Y_{T_0 + 2}, \dots$ treated.
    We assume throughout that these treatment estimates have been appropriately de-meaned or otherwise transformed so that they are exchangeable in absence of a treatment effect.
    This means we can reject the hypothesis of no treatment effect, if we reject the hypothesis of exchangeability:
    \begin{align*}
        H_0 : Y_1, Y_2, \dots, Y_{T_0}, Y_{T_0 + 1}, \dots Y_T \textnormal{ are exchangeable}.
    \end{align*}
    
    In Section~\ref{sec:fixed_effects}, we detail in an interactive fixed effects model how assumptions on the data generating process for difference-in-differences and synthetic control give rise to such exchangeable treatment estimates.

    Moreover, recall our two alternatives: the more specific alternative $H_1^\textnormal{post-exch.}$, where post-treatment data are assumed to be exchangeable among each other, but not exchangeable with the pre-treatment data, and the broader alternative $H_1^\textnormal{generic}$, where the post-treatment data are not exchangeable with the pre-treatment data or among each other.

    \subsection{Coarsening the filtration}
        Interestingly, there exist no non-trivial test martingales for exchangeability under the natural filtration: the filtration under which we have full information about the original data $Y_1, Y_2, \dots, Y_T$ \citep{vovk2021testing, ramdas2022testing}.
        This means it is necessary to consider a less informative `coarsened' filtration.
        The practical consequence is that in exchange for having a powerful test martingale, we are only allowed to base our decision to stop or continue testing on some statistic of the data; not the underlying data itself.

        This may superficially seem paradoxical: if we only base our decision to stop or continue on a statistic of the data, then we must surely lose power?
        This paradox is resolved by realizing the coarsening of the filtration simultaneously reduces the collection of data-dependent (`stopping') times under which our method is required to be anytime-valid.
        This is aptly named `the power of forgetting' by \cite{vovk2023power}.

        Another interpretation is that we may only be interested in alternatives that depend on this statistic, so that the precise values of the underlying data are uninteresting to us.
        This means that we consider a narrower alternative hypothesis towards which we can direct our power.
        In this interpretation, the coarsening of the filtration is inconsequential: it discards irrelevant information.
        
    \subsection{The generic alternative: sequential ranks}\label{sec:seq_ranks}
        For exchangeability, the natural\footnote{Ranks are natural because, if there are no ties, they are in bijection with the group of permutations that underlies exchangeability \citep{koning2024post}. 
        This means that testing uniformity of the ranks can be interpreted as testing uniformity on the group of permutations.} reduction of the filtration is obtained by converting the data into their sequential ranks: the rank $R_t$ of the new observation $Y_t$ among the previous observations $Y_1, \dots, Y_{t-1}$.
        We denote the filtration generated by the sequential ranks by $\mathcal{G} = \{\mathcal{G}_t\}_{t\geq0}$.
        The consequence of this reduction is that we can obtain powerful test martingales, at the cost that our inference can only depend on the sequential ranks.
    
         A key property of the sequential ranks is that under exchangeability, the previous ranks are completely uninformative about the next rank.
        This means that under this reduced filtration, our null hypothesis can be recast as sequentially testing whether the sequential ranks are uniformly distributed \citep{vovk2003testing}:
        \begin{align}\label{eq:hyp_rank}
            H_0^{\text{rank}} : R_t \overset{\text{indep.}}{\sim}\text{Unif}\{1, \dots, t\},
        \end{align}
        for every $t > T_0$. 
        
        This reduces the highly composite hypothesis of exchangeability to a simple hypothesis on the ranks.
        Under this coarsened filtration, the generic alternative hypothesis is simply the negation of the null:
        \begin{align}
             H_1^{\text{generic}} : R_t \overset{\text{indep.}}{\not\sim}\text{Unif}\{1, \dots, t\},
         \end{align}
        which means that the alternative remains composite.
        We handle this in Section \ref{sec:choosing}.
            
        In Theorem \ref{th:rank_test}, we define the building block for our anytime-valid $p$-value for this hypothesis: a sequential $e$-value $e_t$.
        This sequential $e$-value offers great flexibility, as it can be defined for an arbitrary non-negative test statistic $S_t$, that may depend on the past sequential ranks $\mathcal{G}_{t-1}$ and on external randomization.
        An interpretation of this $e$-value is that it quantifies how the current value of the statistic $S_t(R_t)$ compares to the average among its other possible realizations. 

        Moreover, Theorem \ref{th:rank_test} also shows that if the alternative of the $t$'th sequential rank conditional on $\mathcal{G}_{t-1}$ is simple, then log-optimality (see Section \ref{sec:good_martingales}) is attained by choosing $S_t$ proportional to the conditional density of this alternative. 
        All proofs are presented in Appendix \ref{sec:proofs}.
        In Section \ref{sec:choosing}, we appropriate choices for composite alternatives.            
    
        \begin{theorem} \label{th:rank_test}
           Let $\{S_t\}_{t\geq 0}$, $S_t:\{1,...,t\}\rightarrow[0,\infty)$, denote a sequence of functions such that, under $H_0$, $S_t\perp R_t|\mathcal{G}_{t-1}$. 
           Then, for all $t\geq0$,
           \begin{equation}\label{eq:eval}
               e_t = \frac{S_t(R_t)}{\tfrac{1}{t}\sum_{i=1}^t S_t(i)}
           \end{equation}
           is a sequential $e$-value for $H_0$ with respect to filtration $\mathcal{G}$.
           Here, $0/0$ is understood as 1.

           Moreover, for a simple alternative with $g_t|\mathcal{G}_{t-1}$ as the conditional density of $R_t|\mathcal{G}_{t-1}$ with respect to the counting measure on $\{1,\dots,t\}$, the log-optimal test martingale is obtained by choosing a $\mathcal{G}_{t-1}$-measurable test statistic $S_t\propto g_t|\mathcal{G}_{t-1}$.
        \end{theorem}

    \subsection{The post-exchangeable alternative: reduced sequential ranks}\label{sec:reduced_ranks}
        Recall that the post-exchangeable alternative posits that the entire sequence of observations $Y_1, Y_2, \dots, Y_{T_0}, \dots, Y_T$ is not exchangeable, but the post-treatment observations $Y_{T_0+1}, Y_{T_0+2}, \dots$ are.
        For this alternative, we are unconcerned with the ranks of the post-treatment observations $Y_{T_0 +1}, \dots Y_{T}$ amongst themselves: we only care about how the post-treatment observations rank among the pre-treatment observations $Y_1, \dots Y_{T_0}$.
        As a consequence, we can reduce the entire problem to a question concerning what we call the \emph{reduced sequential ranks}.
        \begin{definition}
            The reduced sequential rank $\widetilde{R}_t$ is the rank of $Y_t$, $t > T_0$, among the pre-treatment observations $Y_1, \dots, Y_{T_0}$.
        \end{definition} 
        We call these ranks \emph{reduced}, because they are less informative than the sequential ranks and take value on the reduced space $\{1, \dots, T_0 + 1\} \subseteq \{1, \dots, t\}$, $t > T_0$.
        As we are only interested in these reduced sequential ranks here, we may freely coarsen the filtration to the one induced by these reduced ranks $\{\widetilde{\mathcal{G}}_t\}_{t \geq 0}$.
        
        Under exchangeability, the reduced sequential ranks are generally \emph{not} independent nor uniform like the conventional sequential ranks in \eqref{eq:hyp_rank}. 
        Instead, they follow a conditional categorical distribution over $T_0 + 1$ options, 
        \begin{align} \label{eq:redh0}
            H_0^{\text{red-rank}} : \widetilde{R}_t|\mathcal{G}_{t-1} &\ {\sim}\ \text{Cat} \{q_t^{1}, \dots, q_t^{T_0+1}\},\quad \text{for all } t>T_0.
        \end{align}
        where the probability $q_t^{i}$ that the $t$'th reduced sequential rank $\widetilde{R}_t$ equals $i \in \{1, \dots, T_0 + 1\}$ is proportional to one plus the number of previous sequential ranks that equal $i$:
         \begin{align}
            q_t^i &= \frac{1}{t}\left(1+\sum_{j=T_0+1}^{t-1}I\left[\widetilde{R}_j=i\right]\right),\quad \text{for all } t>T_0.
        \end{align}
        This distribution is uniform on $\{1, \dots, T_0 + 1\}$ at time $t = T_0 + 1$, but generally different at later times, depending on where the previous post-treatment observations rank among the pre-treatment observations.
        As this may be counterintuitive, we provide intuition for this null distribution in Section \ref{sec:vis}.

        For the reduced sequential ranks, the post-exchangeable alternative hypothesis is again composite, and simply states that the reduced sequential ranks do not follow this distribution
        \begin{align}
            H_1^{\text{post-exch.}} : \widetilde{R}_t &\ {\not\sim}\ \text{Cat} \{q_t^{1}, \dots, q_t^{T_0+1}\}.
        \end{align}

         In Theorem \ref{th:red_rank_test}, we present a sequential $e$-value \eqref{eq:evalred} for the reduced sequential ranks that is similar to the one for the sequential ranks from Section \ref{sec:seq_ranks}, but now rescaled using the previously observed reduced ranks. 
         The proof of this theorem is similar to that of Theorem \ref{th:rank_test} and is presented in Appendix \ref{sec:prredranks}.
        \begin{theorem} \label{th:red_rank_test}
           Let $\{\widetilde{S}_t\}_{t\geq T_0}$, $\widetilde{S}_t:\{1,...,T_0+1\}\rightarrow[0,\infty)$, denote a sequence of non-negative functions such that $\widetilde{S}_t\perp \widetilde{R}_t|{\mathcal{G}}_{t-1}$.
           Then, for all $t>T_0$,
            \begin{equation}\label{eq:evalred}
                \tilde{e}_t = \frac{\widetilde{S}_t(\widetilde{R}_t)/q_t^{\widetilde{R}_t}}{\sum_{i=1}^{T_0+1} \widetilde{S}_t(i)}
            \end{equation}
           is a sequential $e$-value for $H_0$ with respect to both filtrations $\widetilde{\mathcal{G}}$ and ${\mathcal{G}}$.
           Here, we use the convention $0/0 = 1$.
           
           Moreover, for a simple alternative $H_1^\text{post-exch.}$ with $\tilde{g}_t|\widetilde{\mathcal{G}}_{t-1}$ as the conditional density of $\widetilde{R}_t|\widetilde{\mathcal{G}}_{t-1}$ with respect to the counting measure on $\{1,\dots,T_0+1\}$, the log-optimal test martingale is obtained by choosing a $\mathcal{G}_{t-1}$-measurable test statistic $\widetilde{S}_t\propto \tilde{g}_t|\widetilde{\mathcal{G}}_{t-1}$.
        \end{theorem}
        
    \subsection{Visualization: sequential ranks and reduced sequential ranks} \label{sec:vis}
        In Figure \ref{fig:redranks}, we illustrate the difference between the null distributions of the sequential ranks and reduced sequential ranks.
        It pictures the null distribution as a histogram of both the rank and reduced rank of the $t = 8$th observation $Y_8$, having observed the same $T_0 = 4$ pre-treatment observations and 3 post-treatment observations.
        In Panel \ref{fig:a}, we see that the exchangeability means that $Y_8$ is equally likely to fall in any slot between pre- or post-treatment observations, and so its rank $R_8$ is uniform on $\{1, \dots, 8\}$.
        In Panel \ref{fig:b}, we picture the null distribution of the reduced sequential rank, which discards the position of $Y_8$ among the post-treatment observations and only conveys its position among the pre-treatment observations.
        
        Comparing the two panels, we see that the mass of the three slots between the first two pre-treatment observations in Panel \ref{fig:a} is 3 out of 8, which exactly corresponds to the mass at slot 2 between the first two pre-treatment observations in Panel \ref{fig:b}.
        More generally, the mass of the null distribution of the reduced rank at a slot between two pre-treatment observations equals one plus the number of post-treatment observations that previously fell into the same slot.
        
        \begin{figure}[h!]
            \centering
            \vspace{0cm}
            \begin{tikzpicture}[scale=0.5]    
                \begin{scope}
                    \foreach \x in {2, 3, ..., 8} {
                        \ifnum\x=3 \fill[black] (\x*2, -1) circle (4.5pt);
                        \else \ifnum\x=4 \fill[black] (\x*2, -1) circle (4.5pt);
                        \else \ifnum\x=6 \fill[black] (\x*2, -1) circle (4.5pt);
                        \else \draw[black] (\x*2, -1) circle (4.5pt);
                        \fi\fi\fi
                    }
                    \foreach \x in {1, 2, ..., 8} {
                        \draw[black] (\x*2 + 0.8, 0.5) rectangle (\x*2 + 1.2, 1.2);
                    }
                    \fill[black] (2*5, -4.1) circle (4.5pt);
                    \node[anchor=west] at (2*5-0.15, -4.45) {\small $Y_8$};
                    \draw[black] (2, 0.5) -- (18.2, 0.5);
            
                    \foreach \x in {1, 2, ..., 8} {
                        \node[scale=0.62] at (\x*2+1, 0) {\x};
                    }
            
                    \foreach \x in {2, 3, 4} {
                        \draw[->, black, bend left] (2*5, -4) to (\x*2 - 1, -1);
                    }
                    \foreach \x in {5, 6} {
                        \draw[->, black] (2*5, -4) to (\x*2 - 1, -1);
                    }
                    \foreach \x in {7, 8, 9} {
                        \draw[->, black, bend right] (2*5, -4) to (\x*2 - 1, -1);
                    }
                \end{scope}
        
                \begin{scope}[xshift=18cm]
                    \foreach \x in {2, 3, ..., 5} {
                        \draw[black] (\x*2, -1) circle (4.5pt);
                    }
                    \fill[black] (2*2+1, -0.75) circle (4.5pt);
                    \fill[black] (2*2+1, -0.4) circle (4.5pt);
                    \fill[black] (2*3+1, -0.75) circle (4.5pt);
                    \fill[black] (2*3+1, -4.1) circle (4.5pt);
                    \node[anchor=west] at (2*3+1-0.1, -4.45) {$Y_8$};
            
                    \foreach \x in {1, 2, ..., 5} {
                        \draw[black] (\x*2 + 0.8, 0.5) rectangle (\x*2 + 1.2, 1.2);
                    }
                    \draw[black] (2*2 + 0.8, 1.2) rectangle (2*2 + 1.2, 1.9);
                    \draw[black] (2*2 + 0.8, 1.9) rectangle (2*2 + 1.2, 2.6);
                    \draw[black] (3*2 + 0.8, 1.2) rectangle (3*2 + 1.2, 1.9);
            
                    \draw[black] (2, 0.5) -- (12.2, 0.5);
            
                    \foreach \x in {1, 2, ..., 5} {
                        \node[scale=0.62] at (\x*2+1, 0) {\x};
                    }
            
                    \foreach \x in {2, 3} {
                        \draw[->, black, bend left] (2*3+1, -4) to (\x*2 - 1, -1);
                    }
                    \foreach \x in {4} {
                        \draw[->, black] (2*3+1, -4) to (\x*2 - 1, -1);
                    }
                    \foreach \x in {5, 6} {
                        \draw[->, black, bend right] (2*3+1, -4) to (\x*2 - 1, -1);
                    }
                \end{scope}
            \end{tikzpicture}
            \begin{subfigure}{0.59\textwidth}
                \centering
                \caption{Null distribution of $R_t$.} \label{fig:a}
            \end{subfigure}
            \hfill
            \begin{subfigure}{0.39\textwidth}
                \centering
                \caption{Null distribution of $\widetilde{R}_t$.} \label{fig:b}
            \end{subfigure}
            \caption{Illustration of the null distribution of the sequential rank $R_t$ (a) and the reduced sequential rank $\widetilde{R}_t$ (b). 
            The circles denote the previous observations: an open circle represents a pre-treatment observation, and a solid circle a post-treatment observation. 
            The bars represent the probability mass of the (reduced) sequential ranks.} \label{fig:redranks}
        \end{figure}
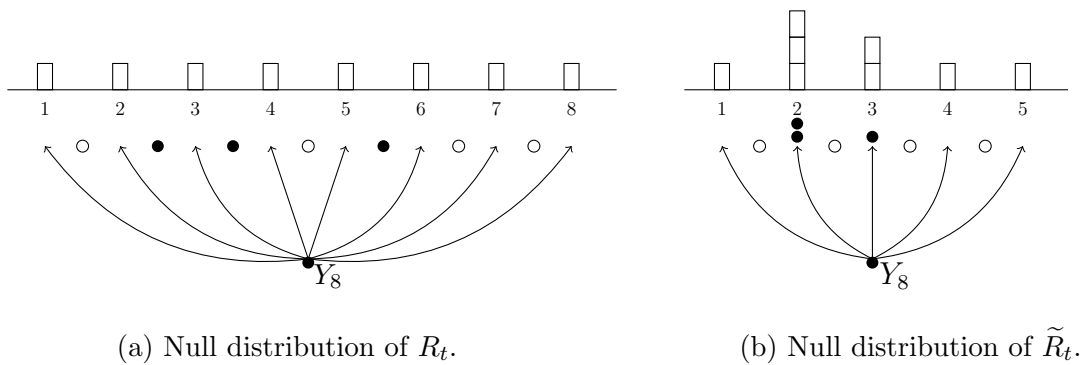

    \section{Choosing a test statistic}\label{sec:choosing}
        We now consider the construction of good test statistics $S_t$ for the alternatives \( H_1^\text{generic} \) and \( H_1^\text{post-exch.} \). 
        Recall from Theorems \ref{th:rank_test} and \ref{th:red_rank_test} that the log-optimal test statistic is proportional to the conditional density of the true data generating distribution under the alternative.
        Unfortunately, this true data generating distribution is unknown as \( H_1^\text{generic} \) and \( H_1^\text{post-exch.} \) are composite, so that we cannot choose the log-optimal test statistic as such.
        To address this problem, we consider two approaches: sequentially learning the distribution under the alternative through a plug-in method, and transforming a (possibly composite) alternative on the outcomes to a single implied distribution on the ranks.

        \subsection{Plug-in method}\label{sec:plugin}
            The idea behind the plug-in method is that if the distribution of the post-treatment ranks is unknown, then perhaps we can sequentially learn it due to our sequential setup.
            The idea is to choose the test statistic $S_t$ to be proportional to (a smoothed version of) the empirical density of the post-treatment ranks that we have observed so-far.
            If the distribution of the post-treatment ranks is stable over time, then after sufficiently many post-treatment draws we should start to learn this distribution.

            The plug-in method we describe here is inspired by the work of \cite{fedorova2012plug}, who use it for testing exchangeability abstractly.
            We modify their approach to accommodate the pre- and post-treatment structure of our problem.
            Moreover, we enhance the plug-in approach for testing against the narrower alternative $H_{1}^\text{post-exch.}$.

            \paragraph*{Rank based test}
                For $H_1^\text{generic}$, we aim to approximate the conditional distribution of the sequential ranks $R_t$, based on $\mathcal{G}_{t-1}$. 
                As the support of the ranks $\{1,\dots,t\}$ is discrete and increases with $t$, we rescale and smooth the sequential ranks before estimating the kernel density.\footnote{Adding independent noise for kernel estimation of discrete data is also called jittering, which is shown to work well in finite samples \citep{nagler2018asymptotic}.} 
                Using independent noise $u_t\sim U(0,1)$ for smoothing, we define the smoothed rescaled ranks as
                \begin{equation} \label{eq:calib}
                    V_t(R_t) = \frac{R_t-u_t}{t},
                \end{equation}
                where $V_t$ has support on $[0,1]$ and is independently uniformly distributed under $H_0^\text{rank}$.
                
                We estimate the distribution of $V_t$ with a continuous kernel density. 
                Similar to \cite{fedorova2012plug}, we apply a standard kernel trick to improve the accuracy of kernel estimates near the boundaries of \([0,1]\).
                Specifically, we estimate the kernel density using an extended sample, defined as
                $\mathcal{V}_{t-1}=\cup_{s=T_0+1}^{t-1}\{-V_s,V_s,2-V_s\}$.
                Let $\widehat{f}_t$ denote the resulting kernel density estimate at time $t$, that is suitably normalized to integrate to 1 on the domain $[0, 1]$ of $V_t$.
                This then yields the plug-in test statistic
                \begin{align}
                    S_t^\text{plug-in}(R_t)
                        = 
                        \widehat{f}_t(V_t(R_t)), \quad \text{for $t>T_0$}, \label{eq:vovkrank}
                \end{align}
                which can be interpreted as the estimated kernel density of the sequential ranks, based on the information in $\mathcal{G}_{t-1}$.
                Since $\widehat{f}_t$ is estimated based on only $R_s,u_s$ for $s<t$, we have that $S_t\perp R_t|\mathcal{G}_{t-1}$. 
                Also, the test statistic is non-negative.
                Hence $S_t$ adheres to all conditions of Theorem \ref{th:rank_test} and can be used for anytime-valid inference.\footnote{In this particular case, the denominator in (\ref{eq:eval}) can be set equal to 1 to decrease variation, while still maintaining validity.}
    
                The log-optimality results imply that if the kernel is a good estimator of the true distribution, then we expect our test martingale to grow quickly.
                \cite{fedorova2012plug} show that this holds asymptotically if the underlying distribution of sequential ranks is \emph{stable}, i.e, the empirical distribution of $\{V_t\}_{t\geq0}$ converges to some distribution with a continuous density.
                Moreover, they show that for a consistent kernel estimator, the asymptotic log growth of the martingale corresponding to $S_t^{\text{plug-in}}$ dominates every martingale based on a test statistic $S_t$ that is constant over time.
            \paragraph*{Reduced rank based test}
                We now propose a specialized improvement of the plug-in method that is more powerful for the post-exchangeable alternative.
                This improvement relies on the fact that the reduced sequential ranks disregard any ordering of post-treatment observations and only retains their position relative to the pre-treatment observations.
                As these post-treatment observations are exchangeable under $H_1^\text{post-exch.}$, the fact that the reduced ranks disregard this information can improve the estimation.
                
                We consider the same kernel density estimator $\widehat{f}_t$, and its CDF $\widehat{F}_t$ as the ones we used for the rank based test.
                We estimate the discrete distribution of $\widetilde{R}_{t}$ by integrating $\widehat{f}_t$ over the regions corresponding to every realization of $\widetilde{R}_t$. 
                Some algebraic manipulation shows that this can be written as the difference of the kernel CDF evaluated at points $\sum_{i=1}^{\widetilde{R}_t-1}q_t^{i}$ and $\sum_{i=1}^{\widetilde{R}_t}q_t^{i}$,
                \begin{equation}\label{eq:vovkrr}
                    \widetilde{S}_t^\text{plug-in}(\widetilde{R}_t)
                        = \widehat{F}\left(\sum_{i=1}^{\widetilde{R}_t}q_t^{i}\right)-\widehat{F}\left(\sum_{i=1}^{\widetilde{R}_t-1}q_t^{i}\right), \quad \text{for }t>T_0.
                \end{equation}
                
                The conditions for anytime-validity as stated in Theorem \ref{th:red_rank_test} hold, as $\widetilde{S}_t^\text{plug-in}$ is non-negative, and independent of $\widetilde{R}_t$ given $\mathcal{G}_{t-1}$.
                The following theorem shows that, under $H_1^{\text{post-exch.}}$, an anytime-valid test based on this $\widetilde{S}^\text{plug-in}$ statistic weakly dominates the sequential $e$-value based on $S^\text{plug-in}$ in terms of the log-target (proof in Appendix \ref{app:prvovk}). 
                Moreover, in simulations we see that it performs substantially better in practice, if the alternative is indeed post-exchangeable.
                
                \begin{theorem} \label{th:vovkdom}
                     Let $e_t$ and $\tilde{e}_t$ denote the sequential $e$-values of the plug-in method based on the sequential ranks (\ref{eq:vovkrank}) and the reduced ranks (\ref{eq:vovkrr}) respectively, based on the same kernel. 
                     Then, for any $\mathbb{Q}\in$ $H_1^{\text{post-exch.}}$, and for all $t>T_0$, 
                    \begin{align}\label{eq:vovkineq}
                        \mathbb{E}^\mathbb{Q}\left[\log e_t|\mathcal{G}_{t-1}\right]\leq \mathbb{E}^\mathbb{Q} \left[\log\tilde{e}_t|{\mathcal{G}}_{t-1}\right].
                    \end{align}
                \end{theorem}

    \subsection{Alternative on the outcome space} \label{sec:gauss}
        Even if we would like to specify a simple distribution on the (reduced) sequential ranks, it can be hard to express such a distribution. 
        Indeed, we may only have an idea about the distribution of the outcomes $Y_1,\dots,Y_T$. 
        Luckily, we can simply take such a distribution on the outcome space and recover the distribution it induces on the (reduced) ranks. 
        Moreover, many distributions on the outcome space may have the same induced distribution on the ranks.
    
        In this section, we illustrate this for a Gaussian alternative. 
        We find that Gaussian distributions with the same effect size / signal-to-noise ratio all have the same implied distribution on the sequential ranks. 
        As a consequence we do not need to specify both the mean and the variance: we only need to choose an effect size. 
    
        While it may feel as a burden that one has to specify such an effect size, note that this is usually also required in the fixed-$T$ setting: either explicitly to plan the number of to-be-gathered observations, or implicitly by gauging whether the number of available observations would be sufficient to find an effect.
    
        \paragraph*{Rank based test}
            Let us consider a simple setting for $\mathbb{Q}$, where the outcomes (pre- and post-treatment) are independently\footnote{The more general exchangeable setting with a correlation coefficient of $\rho$ can also be easily accommodated by dividing $Y_t$ by $\sqrt{(1-\rho^2)}$.} normally distributed, $Y_t \sim \mathcal{N}(\mu_t, \sigma^2)$, with some mean $\mu_t$ and variance $\sigma^2$.
            Suppose the impact of the treatment is a change in the mean of this normal distribution. 
            In particular, $\mu_t = \mu_{\text{pre}}$ for $t \leq T_0$ and $\mu_t=\mu_\text{t,post}$,  otherwise. 

            By Theorem \ref{th:rank_test}, the log-optimal sequential $e$-value is obtained for $S_t\propto f_t$ where
            \begin{align*}
                f_t 
                    = \mathbb{Q}(R_t=r_t|\mathcal{G}_{t-1})
                    =\frac{\mathbb{Q}(R_t=r_t,R_{t-1}=r_{t-1},\dots,R_{T_0+1}=r_{T_0+1})}{\mathbb{Q}(R_{t-1}=r_{t-1},\dots,R_{T_0+1}=r_{T_0+1})},
            \end{align*}
            for all $t>T_0$.
            
            These densities do not permit a simple expression, but they can be easily simulated through Monte Carlo draws of the Gaussian distribution $Y_t\sim \mathcal{N}(\frac{\mu_t-\mu_{\text{pre}}}{\sigma},1)$, where we use that the distribution of the ranks is invariant to univariate shifts in the mean or variance.

            A practical disadvantage is that we must specify a compete path for $\mu_t$, $t > T_0$ under the alternative, which may be difficult.
            The post-exchangeable alternative is exactly the setting where this is not necessary, and we simply fix $\mu_t$ over time.
        
        \paragraph*{Reduced rank based test}
            Under post-exchangeability, the post-treatment observations are exchangeable and so identically distributed. 
            As a consequence, the mean is the same for the entire post-treatment period: $\mu_{t,post}=\mu_{\text{post}}$, for all $t > T_0$.

            Choosing the test statistic $\widetilde{S}_t=\widetilde{f}_t$, where $\tilde{f}_t$ is the distribution of $\widetilde{R}_t|\mathcal{G}_{t-1}$, ensures log-optimality through Theorem \ref{th:red_rank_test}. 
            An efficient way to approximate this density for the Gaussian setting, based on the expected effect size $\frac{\mu_\text{post}-\mu_\text{pre}}{\sigma}$, is provided in Appendix \ref{app:effgauss}.
            The calculation of this $\tilde{f}_t$ is easier than that of $f_t$ as each $\tilde{f}_t$ only requires draws from a $T_0$-dimensional standard normal distribution, instead of a $t$-dimensional distribution.

            In practice, this effect size $ \frac{\mu_\text{post} - \mu_\text{pre}}{\sigma}$ is unknown, but it is possible to average over multiple candidates using the adaptive mixture method described in Appendix~\ref{app:minimax}.
            An appealing property of this weighted averaging is that its cumulative log-growth converges to that of the best-performing candidate in hindsight.
            
            This model averaging approach can be viewed as a computationally tractable, discrete approximation of the mixture method for learning alternatives in the anytime-valid inference literature. 
            While we focus on a finite set of candidates with uniform weighting for simplicity and efficiency, one could also incorporate informative priors or adopt continuous mixture models for greater flexibility.

\section{Illustration: Exchangeability for interactive fixed effects model}\label{sec:fixed_effects}
    In this section, we illustrate how testing for a treatment effect in the interactive fixed effects (IFE) model by \cite{bai2009panel} can be recast to testing for exchangeability. 
    We first present the IFE, and then state conditions under which difference-in-differences (DiD) and the synthetic control method (SCM) generate exchangeable treatment estimates. 
    The exchangeability conditions are identical to the conditions imposed by \cite{abadie2021synthetic} and \cite{chernozhukov2021exact}. 
    If these are not strictly met—for instance, due to serial correlation in treatment estimates—we adopt a block structure, as recommended by \cite{chernozhukov2018exact}.

    \subsection{Interactive fixed effects model}
        The IFE model is a common way to model treatment effects for the synthetic control estimator as it can capture violations of the parallel trend assumption. 
        Consider the following panel for the outcomes of $N+1$ units, $i = 1, \dots, N+1$, of which only the first unit $i = 1$ is treated after time $T_0$,
        \begin{equation} \label{eq:intfixed}
            Y_{it} = \bm{\mu}'_i\bm{\lambda}_t+ \bm{\theta}_t'\bm{Z}_i+ I[i=1 \text{ and }t>T_0]\tau_t+\varepsilon_{it}.
        \end{equation}
        Here, $I[\cdot]$ is the indicator function, $\bm{\mu}_i$ and $\bm{\lambda}_t$ denote the $r$-dimensional factor loadings and time-varying factors, $\tau_t$ is the treatment effect at time $t$, and $\varepsilon_{it}$ is mean zero noise. 
        The $\ell$-dimensional covariates of unit $i$ are denoted with $\bm{Z}_i$ and its corresponding time-varying parameter is denoted with $\bm{\theta}_t$. 
        This specification allows units to have different responses to time trends and thereby violate the parallel trends assumption.

        We follow \cite{abadie2021synthetic} by splitting the pre-treatment observations into a set of \textit{blank} periods $\mathcal{B}=\{1,...,T_{\mathcal{B}}\}$ and a set of \textit{training} periods used for estimation $\mathcal{E}=\{1,...,T_0\}\setminus\mathcal{B}$.
        The training periods are used for estimating the pre-treatment means (DiD) and weights (SCM), and the blank periods are compared to the post-treatment observations. 
        The fact that we do not use these blank periods for estimating weights, helps ensuring exchangeability of the blank- and post-treatment periods under the null.

    \subsection{Inference with counterfactual and synthetic control methods}
        We are interested in conducting inference on the treatment effect $\tau_t$.
        Specifically, we want to test that there is no treatment effect,
        \begin{equation}
            H_0^\text{treatment}: \tau_t=0,\quad \text{ for all } t>T_0 \label{eq:h0}.
        \end{equation}
        In the following two sections, we show under which conditions DiD and SCM generate treatment estimates $\widehat{\tau}_t$, such that $H_0^\text{treatment}$ can be reduced to the hypothesis studied in Sections \ref{sec:seqtest} and \ref{sec:choosing}:
        \begin{align*}
            H_0: \widehat{\tau}_t \text{ is exchangeable over }t \in \mathcal{B} \cup \{T_0+1,...,T\}.
        \end{align*}

        \subsubsection{Difference-in-differences}
            Suppose all units respond equally to changes in the time-varying factors $\bm{\lambda}_t$. In the DiD literature, this is also known as the \emph{parallel trends assumption}. 
            This corresponds to the IFE model with no covariates ($\ell=0$) and, without loss of generality, 
            \begin{equation} \label{eq:partrends}
                \bm{\mu}_i = [\mu_{1i},\mu_2,...,\mu_r]' \quad 
                \bm{\lambda}_t = [1,\lambda_{2t},...,\lambda_{rt}]',\quad  \text{for all $i=1,\dots,N,$ $t=1,\dots,T.$} 
            \end{equation}
            The goal of DiD is to estimate treatment effects for the blank pre-treatment periods $\mathcal{B}$ and the post-treatment observations $\{T_0+1,...,T\}$.
            The DiD estimator does this by filtering out the unit- and time-fixed effects by double differencing,
            \begin{align}
                \widehat{\tau}_t
                    = (Y_{1t}-\bar{Y}_t) - (\bar{Y_1}-\bar{Y}),\quad\text{for all } t \in \mathcal{B} \cup \{T_0+1,...,T\},
            \end{align}
            where $\bar{Y}_t$ is the mean outcome of control units at time $t$, $\bar{Y}_1$ is the average outcome of unit 1 over training periods $\mathcal{E}$, and $\bar{Y}$ is the average outcome of all control units over training period $\mathcal{E}$.
            
            For the IFE model under the parallel trends assumption, the treatment estimator equals
            \begin{equation}\label{eq:diff}
                \widehat{\tau}_t=\tau_t + (\varepsilon_{1t} -\bar{\varepsilon}_t) - (\bar{\varepsilon}_1-\bar{\varepsilon}),
            \end{equation}
            where $\bar{\varepsilon}_t$, $\bar{\varepsilon}_1$, and $\bar{\varepsilon}$ are the respective error terms of $\bar{Y}_t$, $\bar{Y}_1$, and $\bar{Y}$.
            The following proposition states that exchangeability of $\bm{\varepsilon_t}$ implies exchangeability of this DiD estimator under $H_0$.
            
            \begin{proposition} \label{th:DIDth}
                Assume the outcomes $Y_{it}$ follow the IFE model in (\ref{eq:intfixed}) with the parallel trend assumption in (\ref{eq:partrends}). 
                Then, under $H_0^\text{treatment}$, the DiD estimator
                $\{\widehat{\tau}_t\}_{t\in \mathcal{B}\cup \{T_0+1,...,T\}}$ is exchangeable if $\bm{\varepsilon}_t=[\varepsilon_{1t},\dots,\varepsilon_{N+1,t}]'$ is exchangeable over $t=1,...,T$.
            \end{proposition}
            \noindent This proposition is a special case of Proposition \ref{th:scm}, which is presented next in Section \ref{sec:scm_theory}.

        \subsubsection{Synthetic control method} \label{sec:scm_theory}
            In the more general case where the parallel trends assumption is violated, one can use the synthetic control method \citep{abadie2003economic}. 
            This technique attempts to filter out the interactive fixed effects and time-varying parameter $\theta_t$ by constructing a synthetic control unit that approximates the factor loadings and covariates of the treated unit.
            Specifically, it first estimates an $N$-dimensional vector of linear weights $\widehat{\bm{w}}=[\widehat{w}_2,...,\widehat{w}_{N+1}]$ based on the training periods $\mathcal{E}$, and then estimates the synthetic outcomes as $\sum_{i=2}^{N+1} \widehat{w}_iY_{it}$. The treatment effect estimator is then the difference between the outcome of the treated unit and the synthetic outcome,
            \begin{equation}
                \widehat{\tau}_t = Y_{1t}-\sum_{i=2}^{N+1} \widehat{w}_iY_{it}, \quad \text{for all } t \in \mathcal{B} \cup \{T_0+1,...,T\}. \label{eq:tauhat}
            \end{equation}
            
            The SCM weights $\widehat{\bm{w}}$ are usually constructed as follows. 
            Let $\bm{X}_1$ be a $K$-dimensionional vector that contains $K$ characteristics based on the pre-treatment data in $\mathcal{T}$, and collect the same characteristics for the $N$ control units in the $K\times N$ matrix $\bm{X}_0$. 
            For a given non-negative diagonal matrix $\bm{V}$ that determines the importance of these characteristics\footnote{See \cite{abadie2010synthetic} for a data-driven approach to determine $\bm{V}$.}, $\widehat{\bm{w}}$ minimizes the weighted Euclidean distance between the pre-treatment characteristics of the treated and the synthetic control unit,
            \begin{align*}
                \widehat{\bm{w}} = \text{argmin}_{\bm{w}\in \Delta_N} (\bm{X}_1-\bm{X}_0\bm{w})'\bm{V}(\bm{X}_1-\bm{X}_0\bm{w}),
            \end{align*}
            where $\Delta_N$ denotes the $N$-dimensional simplex. 

            To show how a properly chosen $\widehat{\bm{w}}$ can lead to an (approximately) exchangeable treatment estimator, we decompose the treatment effect estimator presented in \eqref{eq:tauhat},

            \begin{align*}
                \widehat{\tau_t} = \tau_t + (\underbrace{\bm{\mu}_1-\sum_{i=2}^{N+1}\widehat{w}_i\bm{\mu}_i}_{\approx 0})\bm{\lambda}_t +(\underbrace{\bm{Z}_1-\sum_{i=2}^{N+1}\widehat{w}_i\bm{Z}_i}_{\approx 0})\bm{\theta}_t+\underbrace{\varepsilon_{1t}-\sum_{i=2}^{N+1}\widehat{w}_i\varepsilon_{it}}_{\text{exchangeable}}.
            \end{align*}

             It is easy to verify whether the covariates are accurately reconstructed by the synthetic control unit ($\bm{Z}_1 \approx \sum_{i=2}^{N+1}\bm{Z}_i\widehat{w}_i$) as the covariates are observable. 
            This is more difficult for the unobserved factor loadings $\bm{\mu}_i$. \cite{abadie2010synthetic} show however, that under mild conditions, when the $K$ characteristics include the pre-treatment outcomes, the weights $\widehat{\bm{w}}^{(T_{\mathcal{E}})}$ asymptotically reconstruct the factor loadings of the treated unit:
            \begin{align*}        \sum_{i=2}^{N+1}\widehat{w}^{(T_{\mathcal{E}})}_i\bm{\mu}_i \xrightarrow{p} \bm{\mu}_1.
            \end{align*}
            
            \noindent Unfortunately, in finite samples, the SCM does not exactly reconstruct the factor loadings and hence we require stronger assumptions for exchangeability than just the exchangeable error term assumption of DiD. 
            The following result restates Theorem 2 of \cite{abadie2021synthetic} for the conditions on exchangeability of the SCM estimator.
            
            \begin{proposition}\label{th:scm}
                Assume the outcomes $Y_{it}$ follow the IFE model in (\ref{eq:intfixed}) where \newline $\{(\bm{\lambda}_t,\bm{\theta}_t)\}_{t\in \mathcal{B}\cup \{T_0+1,...,T\}}$ is exchangeable and independent of $\{\bm{\varepsilon}_t\}_{t\in \mathcal{B}\cup \{T_0+1,...,T\}}$. 
                Also, assume that $\{\bm{\varepsilon}_t\}_{t\in \mathcal{B}\cup \{T_0+1,...,T\}}$ is a sequence of exchangeable random variables. \newline
                Then, under $H_0^\text{treatment}$, $\{\widehat{\tau}_t\}_{t\in \mathcal{B}\cup \{T_0+1,...,T\}}$ is exchangeable.
               The exchangeability assumption of $\{\bm{\theta}_t\}_{t\in \mathcal{B}\cup \{T_0+1,...,T\}}$ can be dropped in the case of perfect pre-treatment fit of the covariates $(\bm{Z}_1 = \sum_{i=2}^{N+1}\bm{Z}_i\widehat{w}_i)$.
            \end{proposition}
        
    \subsection{Block structures for violations of exchangeability}\label{sec:blocks}
        In some cases, size distortions due to violations of the exchangeability conditions can be reduced. 
        If the treatment estimators are expected to be serially dependent, imposing a block-structure on the treatment estimators can reduce some of this dependence \citep{chernozhukov2018exact}. 
        This is done by partitioning the blank periods and post-treatment periods into blocks of $B$ observations (assuming for simplicity that the number of blank periods $T_{\mathcal{B}}$ is divisible by $B$).
        For each block, we then take the block-wise means of the treatment estimators and perform the procedure with these block-wise treatment estimators.
        It is important to note that these blocks also reduce our collection of stopping times; instead of being able to reject after every single observation, we now can only evaluate our test statistic after every $B$ observations.

\section{Simulation study}
    \subsection{Difference-in-differences in a stylized setting}\label{sec:did}
        We start by illustrating the potential benefits of anytime-valid inference in a stylized setting for the DiD estimator.
        A less stylized simulation setup is presented in Section \ref{sec:scm}.
        In our first setting, we assume the conditions of Proposition \ref{th:DIDth} hold, such that the DiD estimator produces exchangeable treatment estimators under $H_0$.
        We compare the anytime-valid tests with the fixed-$T$ test by \cite{abadie2021synthetic} and show how anytime-valid tests are preferred for a wide range of intertemporal preference profiles. 
    
        More specifically, we draw from the IFE model in \eqref{eq:intfixed}, excluding covariates ($\ell=0$) or fixed effects ($r=0$).
        We set $\bm{\varepsilon}_t \overset{\text{i.i.d}}{\sim} N(0,I)$.  
        The anytime-valid tests are based on the reduced ranks, and we use the plug-in and Gaussian statistics as outlined in Sections \ref{sec:plugin} and \ref{sec:gauss}, initializing the plug-in method at $t = T_0 + 1$ with the optimal test statistic of the Gaussian alternative. 
        For now, we correctly specify the effect size of this Gaussian alternative. 
        These anytime-valid tests are compared to the one-sided fixed-$T$ test by \cite{abadie2021synthetic} (see Appendix \ref{app:fixedT}). 
        As the exchangeability conditions hold exactly, we use a block size of $B=1$.           

        \subsubsection{The price paid for anytime-validity}
            We compare the rejection rates of the two anytime-valid tests to that of the not anytime-valid fixed-$T$ test.
            Clearly, these are not direct competitors as the fixed-$T$ test is not anytime-valid, but this comparison allows us to study how much power we lose in exchange for the anytime-validity.

            We consider two variations of the fixed-$T$ test: the \emph{single fixed-$T$ test} where we correctly evaluate the fixed-$T$ $p$-value at a single time $T$, and the \emph{repeated fixed-$T$ test} where we naively apply the fixed-$T$ methodology sequentially by rejecting as soon as the $p$-value dips below $\alpha$.
            The former controls the Type I error but can only be executed \emph{once} at $t=T$, whereas the latter naively applies the fixed-$T$ methodology at every time $t$ but does not control size. 
            
            Figure \ref{fig:comp} illustrates the rejection rates of these tests over time ($\alpha = 0.05$) under $H_0$ and $H_1: \tau_t = 1.5$ for $t > T_0$. 
            The first thing we see in Figure \ref{fig:h0did} is that repeatedly conducting the fixed-$T$ test leads to a large size distortion, reaching up to 20\%.      
            As shown by the `single fixed-$T$' line, applying this test \textit{only} \textit{once} ensures size control.
            Both anytime-valid tests maintain proper size control throughout the post-treatment period.
            We extended this simulation for up to 1000 post-treatment observations, and their size remained below $\alpha = 0.05$.
    
            \begin{figure}[ht!]
                \centering
                \begin{subfigure}[b]{0.45\linewidth}
                    \centering
                    \includegraphics[width=\linewidth]{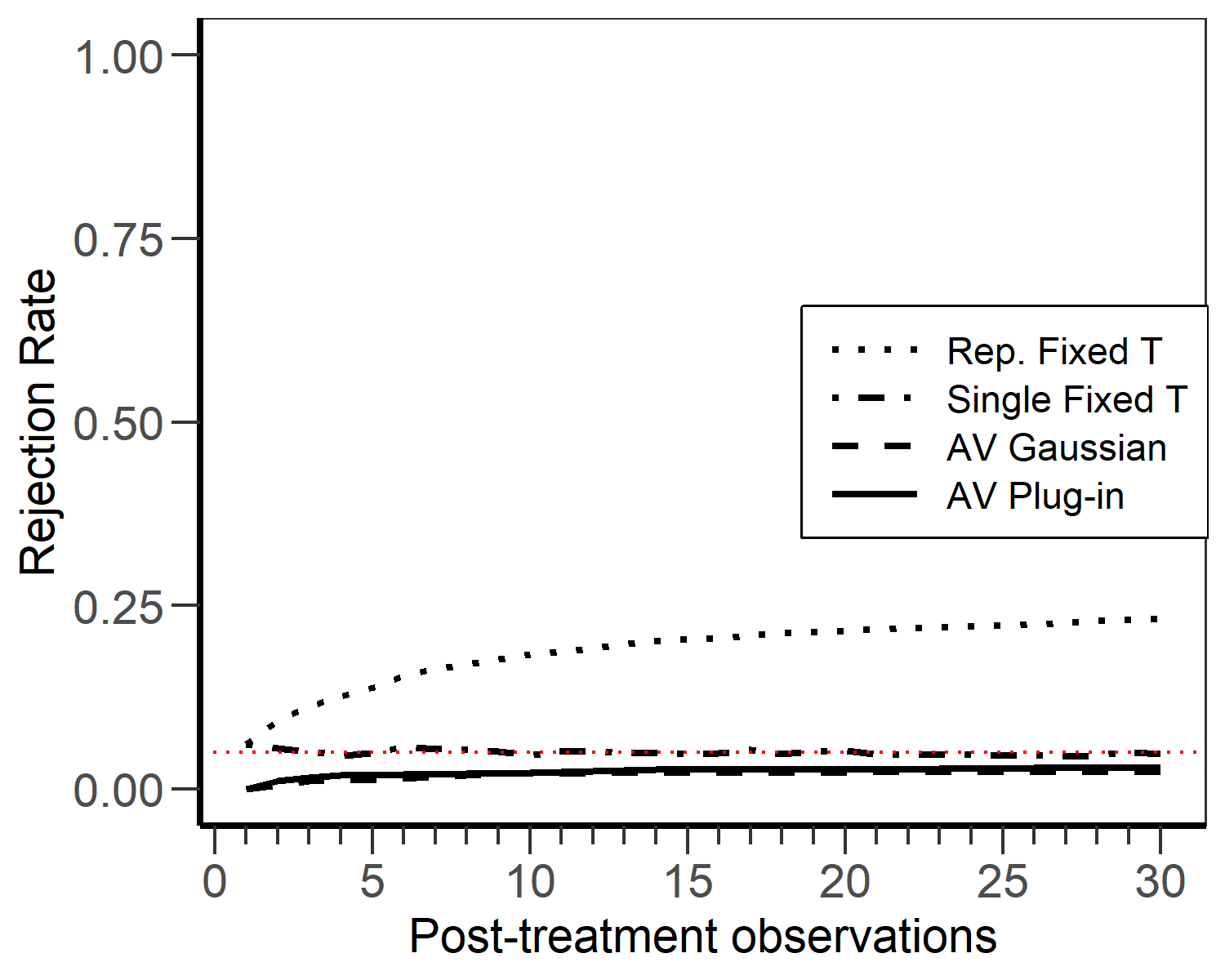}
                    \caption{$H_0: \tau_t=0$}
                    \label{fig:h0did}
                \end{subfigure}%
                \hfill
                \begin{subfigure}[b]{0.45\linewidth}
                    \centering
                    \includegraphics[width=\linewidth]{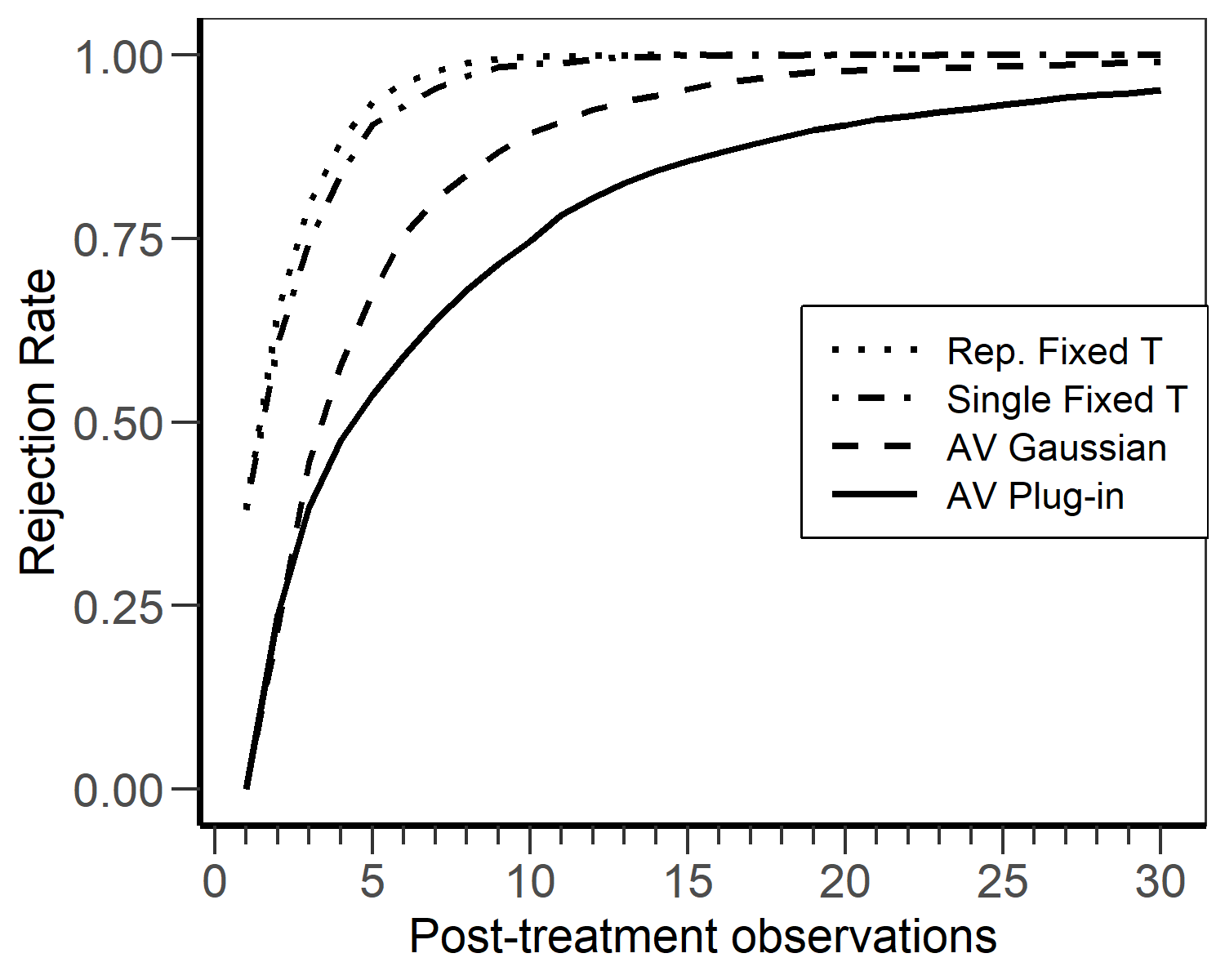}
                    \caption{$H_1: \tau_t=1.5$}
                    \label{fig:h1did}
                \end{subfigure}
                \caption{Rejection rates based on the DiD estimator under $H_0$ and $H_1$ for different testing procedures ($\alpha=0.05$) over time.
                Repeated fixed-$T$: rejection as soon as $p_t$ of fixed-$T$ dips below $\alpha$. 
                Single fixed-$T$: rejection if fixed-$T$ test rejects at time $T$. 
                AV Gaussian: anytime-valid reduced rank test with Gaussian alternative.
                AV Plug-in: anytime-valid reduced rank test with plug-in alternative. 
                Results are based on 2000 simulations from the IFE model with parallel trends and independent standard normal noise. 
                Sample sizes: $T_0=50$, $T_\mathcal{B}=25$, $N=20$. }
                \label{fig:comp}
            \end{figure}
            The rejection rates under the alternative are shown in Figure \ref{fig:h1did}.  
            The power of the repeated fixed-$T$ test should be ignored, as it comes at the costs of high size distortion.
            The power of the single fixed-$T$ method is next highest, but this power can of course only be feasibly attained if we had pre-specified the number of observations to be equal to $T$.
            At a given time $T$, we see that the anytime-valid methods pay a price in terms of power.
            In exchange, a single anytime-valid method attains the displayed power at every time: not just at a single pre-specified time $T$.
            We highlight this difference in Section \ref{sec:shaded}.
            
            Among the anytime-valid tests, the test with Gaussian alternative demonstrates greater power compared to the plug-in method.
            This is expected, as it is log-optimal in this setting, whereas the plug-in method must learn the distribution.
            Both anytime-valid tests converge to full power as the number of post-treatment observations grows.

            \subsubsection{Anytime-valid versus single fixed-\textit{T}} \label{sec:shaded}
                For a fair comparison between the anytime-valid tests and the fixed-$T$ test, we now compare our tests to a \emph{single} fixed-$T$ test. We find that our anytime-valid tests can benefit greatly from being able to reject early or continue testing.
                Figure \ref{fig:measure} presents the rejection rates over time for the anytime-valid test with Gaussian alternative and that of a fixed-$T$ test performed after either 3 (Figure \ref{fig:fix3}) or 6 (Figure \ref{fig:fix6}) post-treatment observations.  
                We highlight differences in rejection rates by shading regions where the anytime-valid test is preferred in dark grey, and the other regions in light grey.
                
                In Figure \ref{fig:fix3}, we study a fixed-$T$ test that is executed after 3 post-treatment observations: $T = T_0 + 3$. 
                We see that most of the advantage of the anytime-valid method stems from rejections occurring after $T$, but it can also benefit from rejecting early (before $T$).  
                The fixed-$T$ test has higher rejection rates only within a narrow band of the number of post-treatment observations.
        
                \begin{figure}[ht!]
                    \centering
                    \begin{subfigure}[b]{0.45\linewidth}
                        \centering
                        \includegraphics[width=\linewidth]{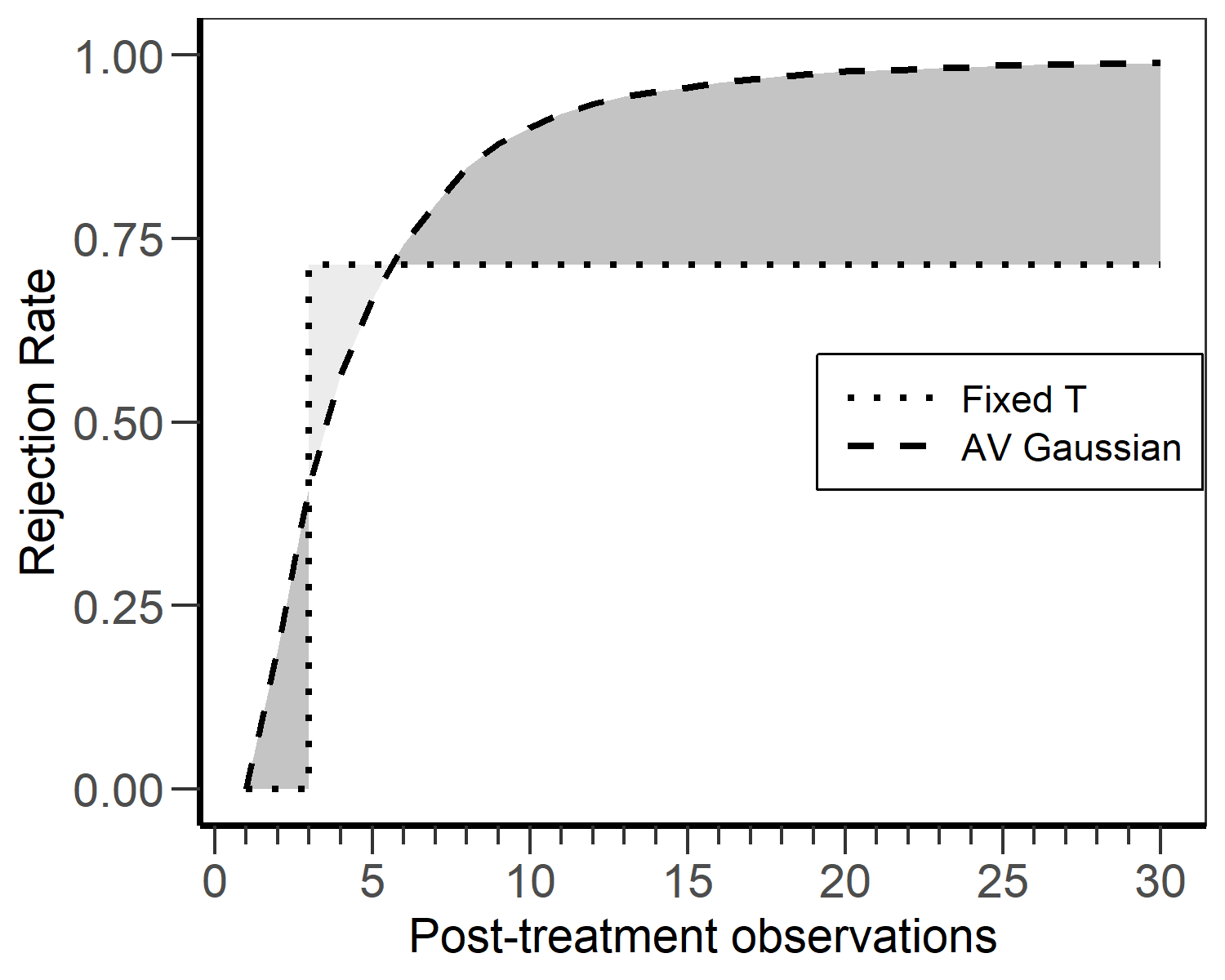}
                        \caption{Testing after 3 observations}
                        \label{fig:fix3}
                    \end{subfigure}%
                    \hfill
                    \begin{subfigure}[b]{0.45\linewidth}
                        \centering
                        \includegraphics[width=\linewidth]{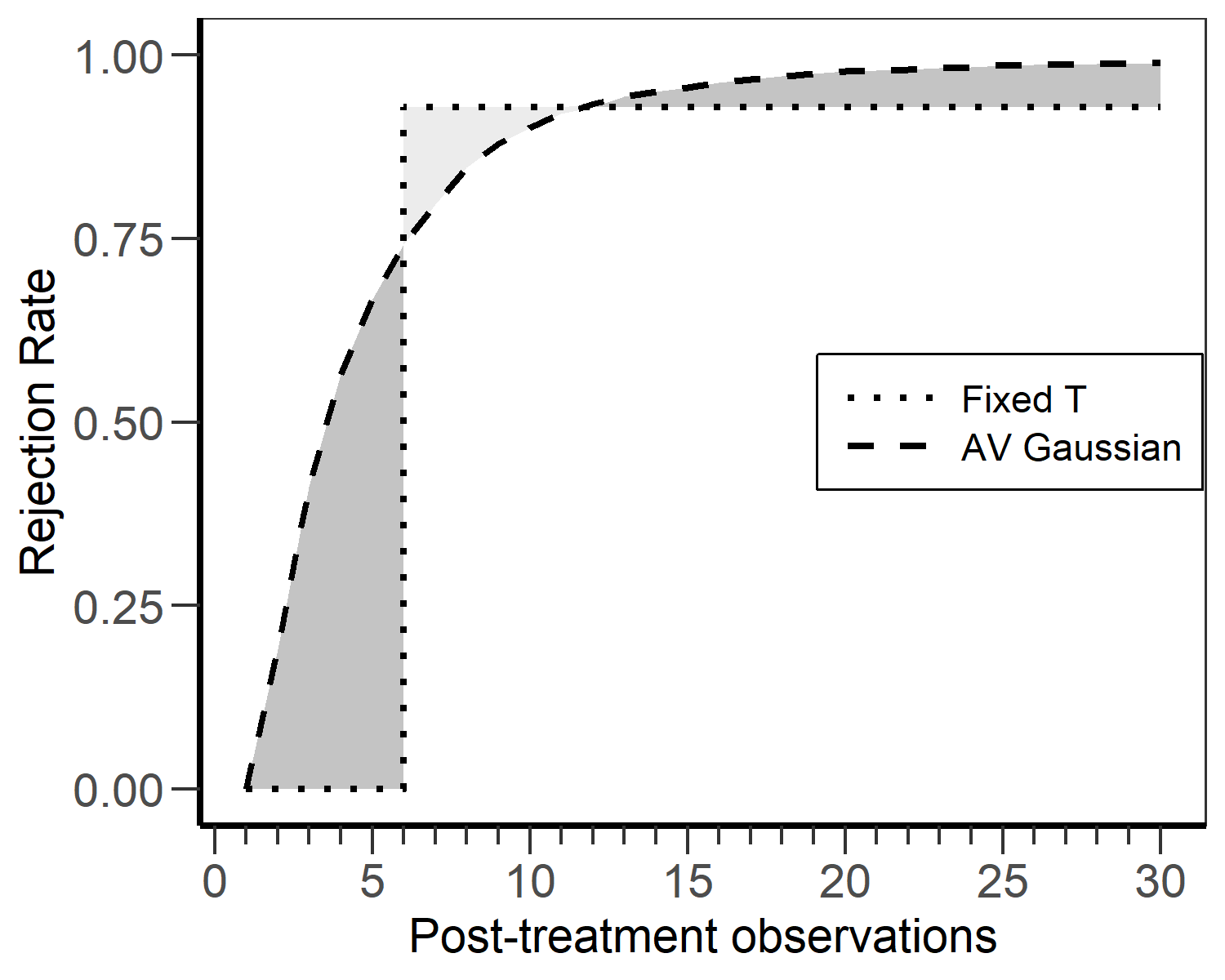}
                        \caption{Testing after 6 observations}
                        \label{fig:fix6}
                    \end{subfigure}
                    \caption{Differences in power between the anytime-valid test with Gaussian alternative and a single fixed-$T$ test conducted after either 3 or 6 post-treatment observations. 
                    Shaded areas represent differences in rejection rates over time. 
                    Simulation settings are the same as in Figure \ref{fig:comp}.}
                    \label{fig:measure}
                \end{figure}
            
                In Figure \ref{fig:fix6}, we make the same comparison for a fixed-$T$ test that rejects after 6 post-treatment observations. 
                Most of the benefit of the anytime-valid test now comes from the ability to reject earlier than the fixed-$T$ test, though some advantage remains for rejections occurring later. 
                
                Comparing the volume of the dark gray and light gray regions in Figure \ref{fig:comp} suggests that this anytime-valid test offers some intertemporal advantage over the fixed-$T$ tests.  
                We study this in more detail in the following section.
                
            \subsubsection{Intertemporal preferences}
                In practice, one may have intertemporal preferences due to, the (opportunity) costs involved with conducting an experiment. 
                To capture these preferences over time, we consider a simple discounted version of power.
                
                Suppose $H_0$ is false, and a decision-maker derives utility from knowing this at $t > T_0$.  
                Consider a normalized discounted utility model, where the expected utility $U_S$ of test $S$ is a discounted sum of its rejection probabilities:
                
                \begin{equation}
                    \mathbb{E}[U_S] = \sum_{t=T_0+1}^T \delta^t P[\text{$H_0$ rejected by test $S$ at $t' \leq t$}], \label{eq:utility}
                \end{equation}
                for some scalar discount factor $\delta \in (0,1]$.  
                When $\delta=1$, the difference in utility between the anytime-valid and fixed-$T$ is simply the difference in volume of the dark grey area and the light grey area in Figure \ref{fig:comp}.
                A smaller $\delta$ implies that the volume of earlier shaded regions is weighted more heavily.
                One economic interpretation of $\delta$ is that it represents the (opportunity) costs of continuing the study.
                A smaller $\delta$ suggests that future rejections carry higher costs, making them less important for the decision-maker.
        
                As seen in Section \ref{sec:shaded}, the intertemporal advantage of an anytime-valid test over a fixed-$T$ test varies with its rejection time $T$.
                We now demonstrate that for a wide range of discount factors $\delta$, the anytime-valid tests are preferred over almost every possible fixed-$T$ test.
                In practice, the $T$ for which the fixed-$T$ test attains the highest discounted utility is unknown a priori, further supporting the use of anytime-valid tests. 
                
               In Figure \ref{fig:delta_did}, we compare the utility of different fixed-$T$ tests with the utility of the anytime-valid tests.
               Each $x$-axis tick corresponds to a fixed-$T$-test executed \textit{only} at that specific time $T$.
               We shade the regions of discount factors $\delta$ where fixed-$T$ tests performed at time $T$ are preferred over the anytime-valid tests.  
               Figure \ref{fig:deltagaussdid} shows that for $\delta \geq 0.8$, the anytime-valid test with Gaussian alternative dominates \emph{all} fixed-$T$ tests. Hence, for a (mildly) patient decision-maker, there does not exist any $T$ at which it would have been better to perform a fixed-$T$ test.
               For low discount factors, the fixed-$T$ test with few post-treatment observations is preferred.
               
               In contrast, Figure \ref{fig:deltavovkdid} shows that the plug-in test is outperformed by a substantial range of fixed-$T$ tests for all values of $\delta$. However, as opposed to the fixed-$T$ test, this plug-in method requires no prior knowledge on the effect size. This suggests that in this setting, the plug-in method is only beneficial when there is high uncertainty on the actual effect size.

                \begin{figure}[ht!]
                    \centering
                    \begin{subfigure}[b]{0.49\linewidth}
                        \centering
                        \includegraphics[width=\linewidth]{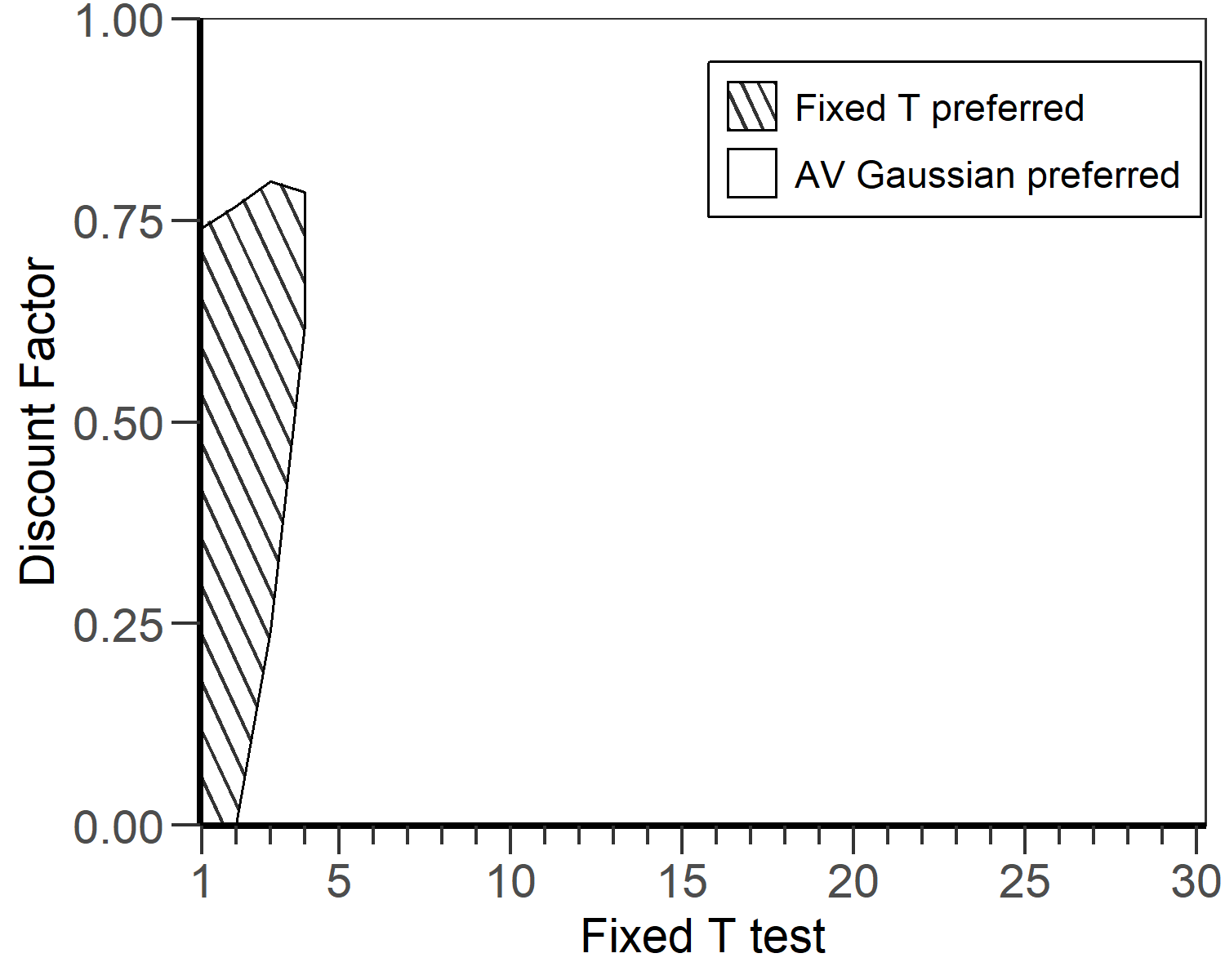}
                        \caption{Anytime-valid test with Gaussian alternative vs. range of fixed-$T$ tests}
                        \label{fig:deltagaussdid}
                    \end{subfigure}%
                    \hfill
                    \begin{subfigure}[b]{0.49\linewidth}
                        \centering
                        \includegraphics[width=\linewidth]{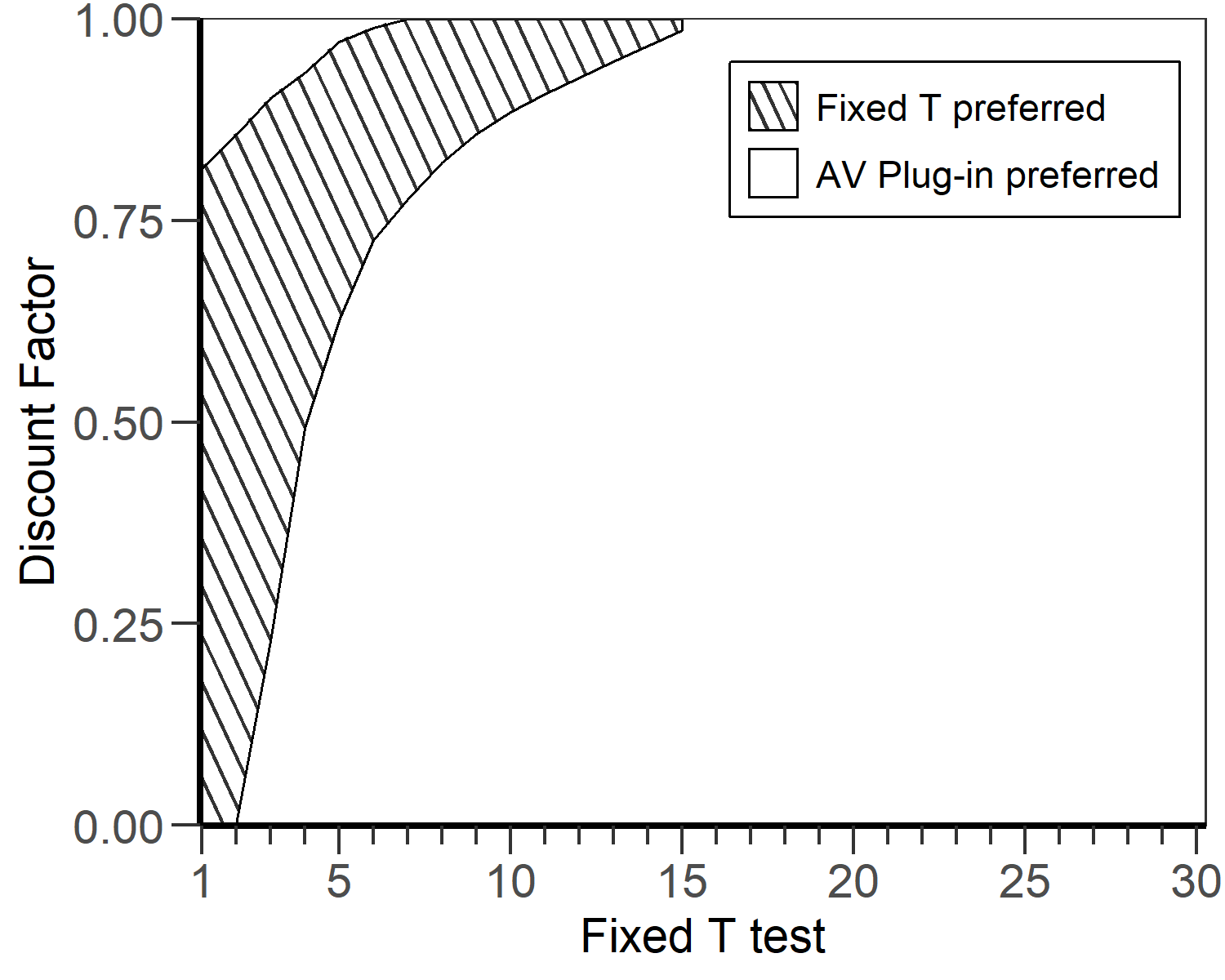}
                        \caption{Anytime-valid test with plug-in alternative vs. range of fixed-$T$ tests}
                        \label{fig:deltavovkdid}
                    \end{subfigure}
                    \caption{Regions of discount factors $\delta$ in which the fixed-$T$ test performed at different levels of $T$ achieves higher discounted utility, as defined in (\ref{eq:utility}), than anytime-valid tests.
                    Simulation settings are the same as in Figure \ref{fig:comp} with $H_1: \tau_t = 1.5$.}
                    \label{fig:delta_did}
                \end{figure}

    \subsection{Synthetic control method and misspecification} \label{sec:scm}
        In Section \ref{sec:did}, we considered a stylised setting, in which the treatment estimates were exactly exchangeable.
        In this setting, we illustrate our methodology on SCM and simultaneously demonstrate the robustness under possible misspecification. 
        We show how size distortion, caused by violations of the conditions in Proposition \ref{th:scm}, can be effectively reduced by employing a block structure. 
        In addition, we consider a setting where the anytime-valid test with Gaussian alternative is based on the incorrect effect size specification. 
        There, we show how adaptive mixtures across multiple alternatives can enhance power. 
        Lastly, we also consider a setting with a dynamic treatment effect, which has a varying strength over time.
    
        In this setup, we consider 3 time-varying factors ($r=3$) and generate the time-varying factors and errors through a VAR(1) process defined as follows:
        \begin{align}
            \bm{\lambda}_t &= \rho_\lambda\bm{\lambda}_{t-1}+ \bm{\xi}_t, & \bm{\xi}_t &\sim N(0,(1-\rho_\lambda^2)I_3)\label{eq:scm1} \\
            \bm{\varepsilon}_t &= \rho_{\varepsilon}\bm{\varepsilon}_{t-1} + \bm{\nu}_t, & \bm{\nu}_t &\sim N(0,\sigma^2(1-\rho_\varepsilon^2)I_{N+1})\label{eq:scm2}
        \end{align}
        where $\rho_\lambda, \rho_\varepsilon \in (-1,1)$ are the autoregressive parameters of $\bm{\lambda}_t$ and $\bm{\varepsilon}_t$, and we set $\sigma^2=1$.
        Proposition \ref{th:scm} implies that the tests are exactly valid when $\rho_\lambda = \rho_\varepsilon = 0$. 
        Since the rejection rates over time and intertemporal preferences in this exactly valid case are similar to the DiD setting (see Appendix \ref{app:similar}), we immediately turn to the effect of non-exchangeability on the size.
    
    \subsubsection{Size distortion due to non-exchangeability}
        We examine the size distortion of three tests: a fixed-$T$ test carried out after 12 post-treatment observations, and two reduced-rank anytime-valid tests based on 30 post-treatment observations. 
        The sample of the anytime-valid tests is deliberately set larger as its size control should also hold after $T=12$. 
        
        Table \ref{tab:b3} presents the size of these tests for different sample dimensions and values of $\rho_\lambda$ and $\rho_\varepsilon$, imposing a block-structure ($B=3$). 
        We also consider the case where $B=1$ in \ref{app:blocks}, where we find settings with high size distortions, but these are similar across fixed-$T$ and anytime-valid tests\footnote{As we impose a block structure of $B$, we rescale the pre- and post-treatment sample lengths with $B$ to ensure a fair comparison between $B=1$ and $B=3$.}.
    
        As expected, when the autoregressive components are equal to 0 ($\rho_\lambda = \rho_\varepsilon = 0$), all tests maintain size control across all sample dimensions. 
        The anytime-valid tests are slightly undersized for this finite horizon. This is not surprising, as they are valid for any final horizon $T$.
        When the autoregressive component of the time-varying factors $\rho_\lambda$ increases, we observe larger sizes for all three tests. 
        These distortions are largest when $\rho_\lambda$ is large and the pre-treatment sample size is small ($T_0 = 20$). 
        In these small samples, the SCM estimator does not adequately reconstruct the factor loadings of the treated unit, leading to non-exchangeability of the SCM estimates. 
        When $T_0$ is large,  most tests exhibit relatively low size distortion. 
        This especially holds for the Gaussian anytime-valid test (see column 6 of Table \ref{tab:b3})
        
        When the idiosyncratic serial correlation ($\rho_\varepsilon$) grows, the size increases, regardless of the dimension of the pre-treatment sample. 
        This is because the SCM only filters out variation that is shared across units, and these idiosyncratic shocks are unit-specific by construction.
        In practice, one often assumes that the idiosyncratic components $\varepsilon_{it}$ have some common variation across units and their serial dependence can therefore be captured by the unobserved time-varying factors.
        
        Overall, the size distortions observed for the fixed-$T$ test and the anytime-valid tests do not differ substantially.
        However, between the two anytime-valid tests, the plug-in method appears less robust than the Gaussian alternative.
        An explanation for this is that the flexibility of this plug-in method also makes it more sensitive to violations of exchangeability.
        As the plug-in method adaptively learns the distribution of the ranks, it gradually learns in what way the exchangeability is violated and attempts to exploit this.
    
        \setlength{\tabcolsep}{3.5pt}
        \renewcommand{\arraystretch}{0.5}
        \begin{table}[h!]
\centering
\caption{Rejection rates under $H_0$ $(\alpha=0.05)$ with $B=3$} \label{tab:b3}
\small
\begin{tabular}{@{}cc|ccccccccccc@{}}
\toprule
\midrule
&&\multicolumn{9}{c}{$\bm{\rho_\varepsilon=0}$} \\
\midrule
&& \multicolumn{3}{c}{Fixed-$T$} & \multicolumn{3}{c}{AV Gaussian} & \multicolumn{3}{c}{AV Plug-in} \\ 
\cmidrule(lr){3-5} \cmidrule(lr){6-8} \cmidrule(lr){9-11}
 $\rho_\lambda$ & $N$ & $T_0=20$ & $T_0=50$ & $T_0=100$ & $T_0=20$ & $T_0=50$ & $T_0=100$ & $T_0=20$ & $T_0=50$ & $T_0=100$  \\
\midrule
\multirow{2}{*}{0} &20 & 0.06 & 0.06 & 0.05 & 0.02 & 0.03 & 0.02 & 0.02 & 0.02 & 0.02 \\ 
 &50 & 0.05 & 0.05 & 0.05 & 0.02 & 0.04 & 0.02 & 0.02 & 0.02 & 0.02 \\ 
\midrule
\multirow{2}{*}{0.25} &20 & 0.06 & 0.06 & 0.06 & 0.03 & 0.03 & 0.02 & 0.02 & 0.02 & 0.02 \\ 
 &50 & 0.06 & 0.05 & 0.05 & 0.02 & 0.04 & 0.02 & 0.02 & 0.03 & 0.02 \\ 
\midrule
\multirow{2}{*}{0.5} &20 & 0.07 & 0.06 & 0.06 & 0.04 & 0.04 & 0.03 & 0.04 & 0.03 & 0.03 \\ 
 &50 & 0.06 & 0.06 & 0.05 & 0.03 & 0.03 & 0.02 & 0.03 & 0.03 & 0.03 \\ 
\midrule
\multirow{2}{*}{0.75} &20 & 0.10 & 0.10 & 0.09 & 0.05 & 0.05 & 0.03 & 0.10 & 0.07 & 0.06 \\ 
 &50 & 0.08 & 0.08 & 0.08 & 0.05 & 0.05 & 0.03 & 0.07 & 0.05 & 0.05 \\ 
\midrule
$\rho_\lambda$ & $N$ & \multicolumn{9}{c}{$\bm{\rho_\varepsilon=0.25}$} \\ 
\midrule
\multirow{2}{*}{0} &20 & 0.06 & 0.07 & 0.07 & 0.03 & 0.04 & 0.03 & 0.03 & 0.04 & 0.05 \\ 
 &50 & 0.07 & 0.06 & 0.07 & 0.03 & 0.05 & 0.02 & 0.03 & 0.03 & 0.03 \\ 
\midrule
\multirow{2}{*}{0.25} &20 & 0.07 & 0.07 & 0.06 & 0.04 & 0.04 & 0.03 & 0.04 & 0.04 & 0.03 \\ 
 &50 & 0.07 & 0.08 & 0.08 & 0.03 & 0.04 & 0.03 & 0.04 & 0.03 & 0.04 \\ 
\midrule
\multirow{2}{*}{0.5} &20 & 0.08 & 0.08 & 0.07 & 0.04 & 0.04 & 0.03 & 0.05 & 0.04 & 0.04 \\ 
 &50 & 0.08 & 0.07 & 0.06 & 0.03 & 0.04 & 0.03 & 0.05 & 0.04 & 0.04 \\ 
\midrule
\multirow{2}{*}{0.75} &20 & 0.13 & 0.11 & 0.10 & 0.07 & 0.05 & 0.04 & 0.10 & 0.08 & 0.08 \\ 
 &50 & 0.10 & 0.09 & 0.08 & 0.06 & 0.05 & 0.03 & 0.07 & 0.07 & 0.06 \\ 
\midrule
$\rho_\lambda$ & $N$ & \multicolumn{9}{c}{$\bm{\rho_\varepsilon=0.5}$} \\ 
\midrule
\multirow{2}{*}{0} &20 & 0.10 & 0.10 & 0.09 & 0.05 & 0.04 & 0.03 & 0.06 & 0.06 & 0.07 \\ 
 &50 & 0.11 & 0.10 & 0.09 & 0.05 & 0.05 & 0.03 & 0.06 & 0.07 & 0.06 \\ 
\midrule
\multirow{2}{*}{0.25} &20 & 0.09 & 0.10 & 0.09 & 0.05 & 0.05 & 0.03 & 0.08 & 0.06 & 0.07 \\ 
 &50 & 0.10 & 0.10 & 0.10 & 0.04 & 0.05 & 0.03 & 0.07 & 0.07 & 0.07 \\ 
\midrule
\multirow{2}{*}{0.5} &20 & 0.10 & 0.09 & 0.09 & 0.05 & 0.04 & 0.04 & 0.07 & 0.07 & 0.07 \\ 
 &50 & 0.10 & 0.11 & 0.10 & 0.05 & 0.04 & 0.03 & 0.06 & 0.07 & 0.08 \\ 
\midrule
\multirow{2}{*}{0.75} &20 & 0.12 & 0.13 & 0.11 & 0.09 & 0.06 & 0.05 & 0.13 & 0.10 & 0.11 \\ 
 &50 & 0.12 & 0.11 & 0.11 & 0.06 & 0.06 & 0.04 & 0.10 & 0.10 & 0.10 \\ 
\midrule
\bottomrule
\end{tabular}
\begin{tablenotes} 
    \item\textit{Notes:} 2000 repetitions of simulation design from IFE model as described in (\ref{eq:scm1},\ref{eq:scm2}). Fixed-$T$: test by \cite{abadie2021synthetic}, AV Gaussian: Anytime-valid reduced rank test with Gaussian alternative, AV Plug-in: Anytime-valid reduced rank test with plug-in alternative. The fixed-$T$ test is performed after $12B=36$ post-treatment observations. $T_\mathcal{B}=T_0/2$ and $T-T_0=30B=90$.
\end{tablenotes}
\end{table}

        \subsubsection{Power of anytime-valid Gaussian method under mis-specification}\label{sec:mispec}
            In the previous simulations, we used the true effect size for the Gaussian anytime-valid method.
            We now study the sensitivity of the power to a possible misspecification of the effect size. 
            The key finding is that the power can vary quite strongly if the effect size is misspecified, and that this can be overcome by taking a mixture over different plausible effect sizes.

            Using the same IFE setup with $\rho_\lambda=\rho_\varepsilon=0$, Figure \ref{fig:diffpriors} shows the power of the anytime-valid test with Gaussian alternative when the true effect size is off by a factor $c > 0$. 
            For $c > 1$, this means that we test against an effect size that is larger than the true effect size.
            For $c < 1$, we test against an effect size that is smaller than the true one.
            
            We find that if $c > 1$, the rejection rate is higher at the beginning, but it grows less quickly later on. 
            The tests with smaller $c$ start out less powerful, but grow more powerful later on. 
            This illustrates that different $c$ correspond to different patience preferences.\footnote{\cite{koning2024continuous} discusses how in a Gaussian location model, these different effect sizes correspond to optimal $e$-values for different power targets.}
    
            \begin{figure}[ht!]
            \centering
            \begin{subfigure}[b]{0.45\linewidth}
                \centering
                \includegraphics[width=\linewidth]{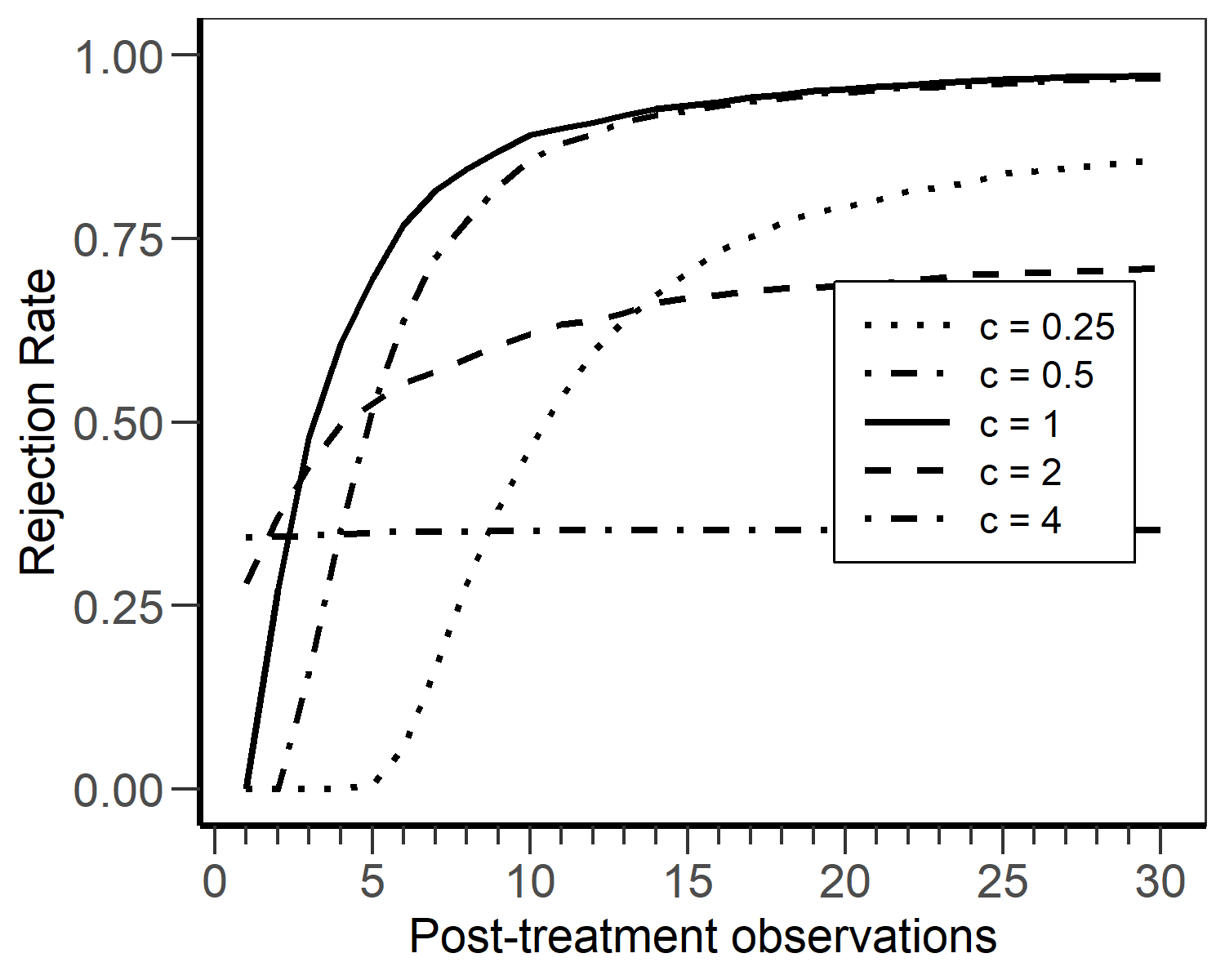}
                \caption{Log-optimal vs. misspecified Gaussians }
                \label{fig:diffpriors}
            \end{subfigure}%
            \hfill
            \begin{subfigure}[b]{0.45\linewidth}
                \centering
                \includegraphics[width=\linewidth]{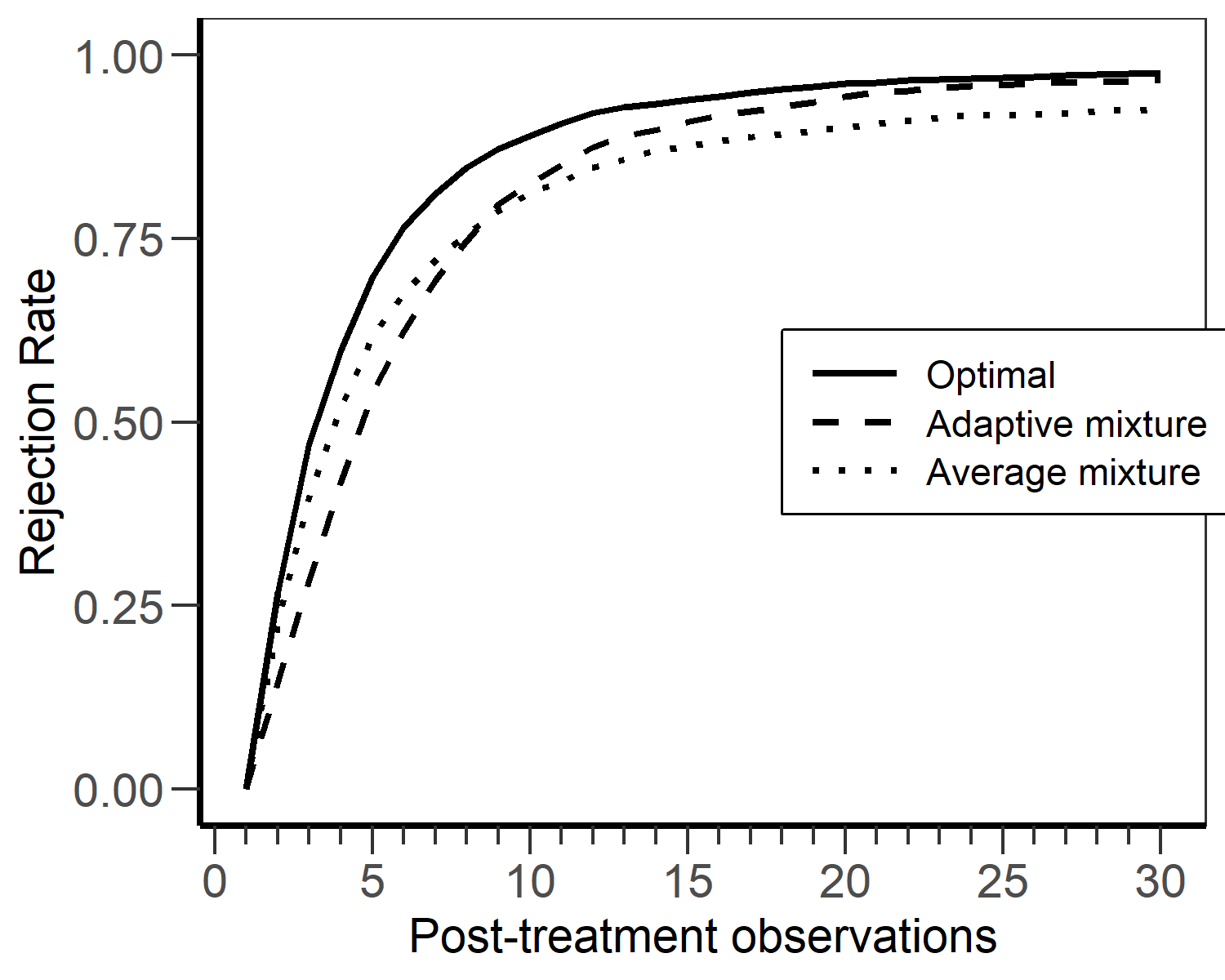}
                \caption{Log-optimal vs. mixtures}
                \label{fig:mixing}
            \end{subfigure}
            \caption{Rejection rates based on the SCM estimator under $H_1:\tau_t=2$ of different testing procedures ($\alpha=.05$) over time. 
            We draw 2000 simulations from the IFE with $\rho_\varepsilon=\rho_\lambda=0$ and independent standard normal noise. 
            The sample dimensions are $T_0=50,T_\mathcal{B}=25,N=20$.}
            \label{fig:priors}
        \end{figure}
            
            In practice, the actual effect size is often unknown. 
            Let $C$ denote a countable set of candidate $c$'s. 
            Proposition \ref{th:mixing} in Appendix \ref{app:minimax} implies that if $1\in C$, then the average over the martingales implied by $C$ has asymptotic log-optimal growth. 
            Since this average is actually an adaptive weighted average over $e$-values, we call this strategy the \textit{adaptive mixture}. 
            Another possibility is to take a uniform average over all the candidate $e$-values and constructing a martingale through their running product. 
            We call this strategy the \textit{average mixture}.
        
            We compare our mixture strategies, with $C=\{0.25,0.5,1,2,4\}$, to the log-optimal rejection rates in Figure \ref{fig:mixing}. 
            Both weighting strategies are less powerful than the log-optimal strategy, which is explained by the inclusion of the sub-optimal martingales. 
            After $T=10$, the adaptive mixture starts outperforming the average mixture and its power approaches that of the log-optimal strategy.

        \subsubsection{Power under dynamic treatment effects}
            In practice, the treatment effect may be dynamic, with which we mean that the treatment effect may change over time. 
            In this section, we consider a setting where the treatment effect grows linearly over time $\tau_t=1 +\frac{t-T_0}{15}$.
            We compare the fixed-$T$ test, plug-in method with sequential ranks, and an \textit{adaptive mixture} (Section \ref{sec:mispec}) over reduced rank anytime-valid Gaussian tests with $C=\{0.25,0.5,1,2,4\}$. 
            In short, the results show that although the anytime-valid methods remain effective, their advantage over the fixed-$T$ test is slightly reduced.  
        
            Figure \ref{fig:dyn_delta} presents the intertemporal preference plots for both of the anytime-valid tests. 
            The anytime-valid test with Gaussian alternative is still preferred over almost all fixed-$T$ tests for discount factors $\delta>0.9$. 
            This is slightly weaker than the constant treatment effect, where this held for all fixed-$T$ tests. 
            The plug-in test statistic seems to perform less well.
            It is dominated by a large set of fixed-$T$ tests for every discount factor $\delta$. 
            An explanation for this is that the distribution of the ranks is not \textit{stable}, which makes it difficult to effectively learn the distribution of the ranks.
            One important note is that the inherent uncertainty regarding the dynamics of the treatment effect simultaneously make it more difficult to specify at which $T$ to test for a fixed-$T$ test. 
            Hence, the increased range for which the fixed-$T$ is preferred over these anytime-valid tests compared to the instantanuous setting is partially offset by the additional insecurity in specifying $T$.
            
            \begin{figure}[ht!]
                \centering
                \begin{subfigure}[b]{0.45\linewidth}
                    \centering
                    \includegraphics[width=\linewidth]{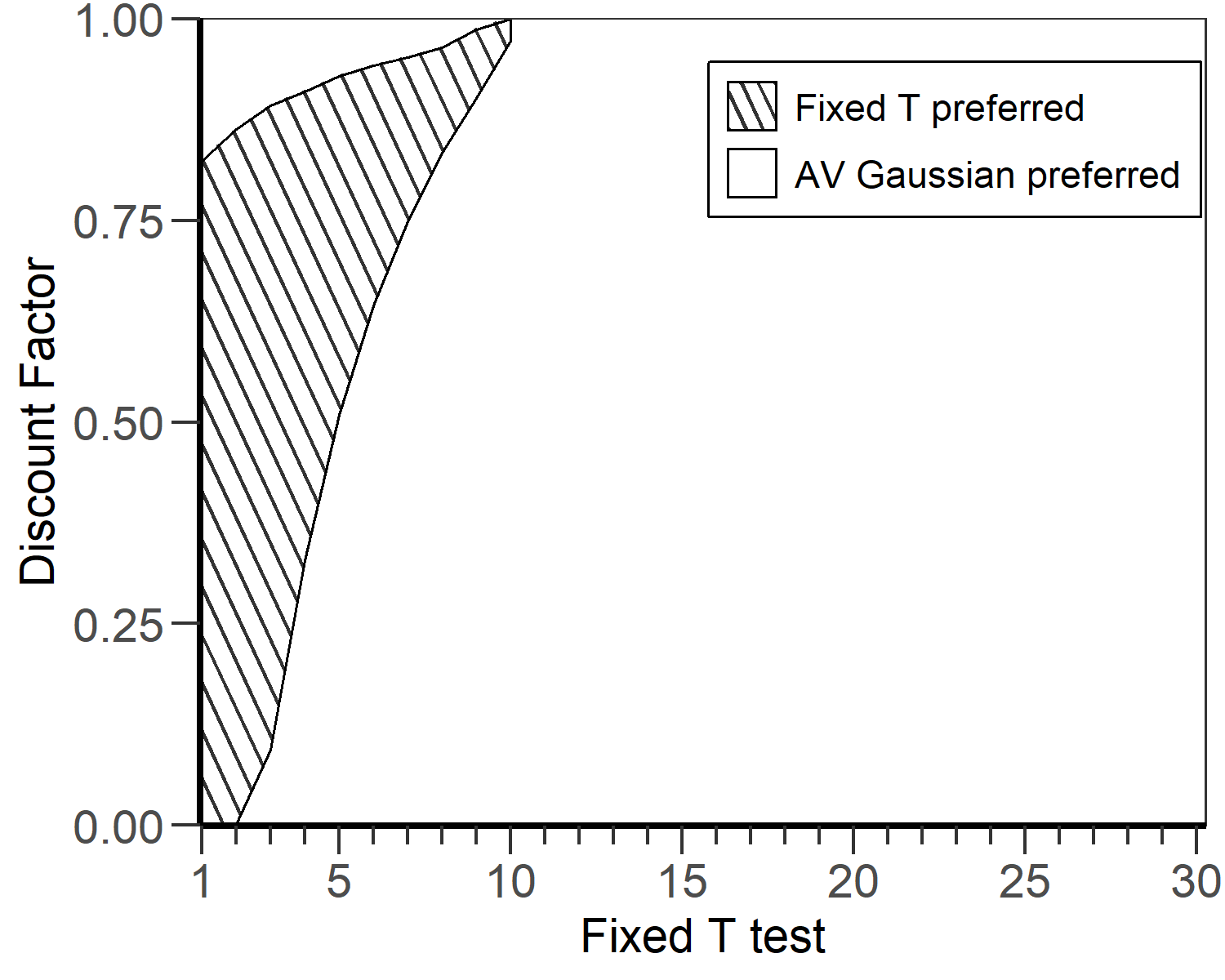}
                \caption{Adaptive mixture of anytime-valid Gaussian tests vs. range of fixed-$T$ tests}
                    \label{fig:dyn_gauss}
                \end{subfigure}%
                \hfill
                \begin{subfigure}[b]{0.45\linewidth}
                    \centering
                    \includegraphics[width=\linewidth]{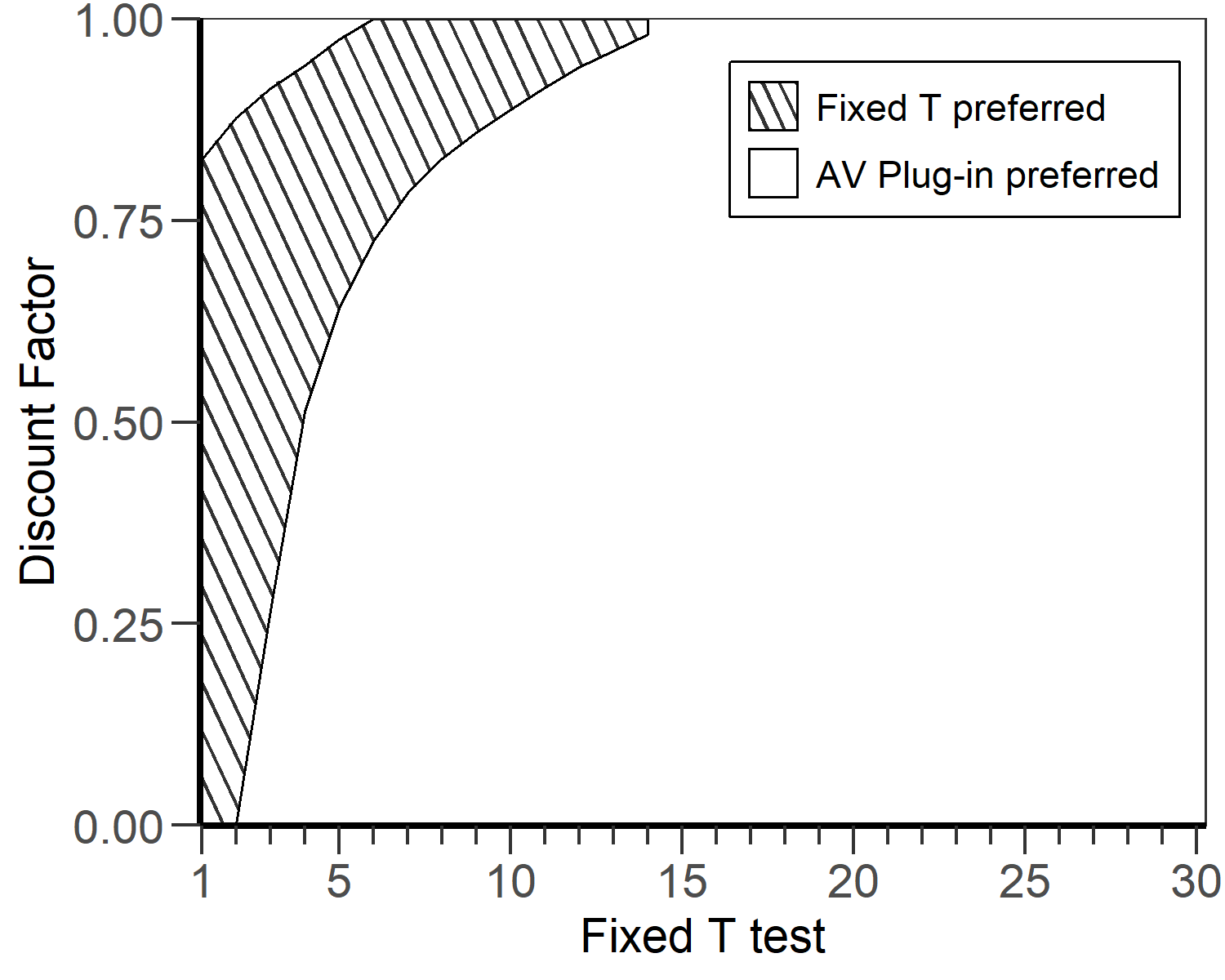}
                    \caption{Anytime-valid test with plug-in alternative vs. range of fixed-$T$ tests}
                    \label{fig:dyn_vovk}
                \end{subfigure}
                \caption{Regions of discount factors $\delta$ in which the fixed-$T$ test at different levels of $T$ attains higher discounted utility than the anytime-valid tests. 
                Test is based on the SCM estimator under $H_1:\tau_t=1 +\frac{t-T_0}{15}$. 
                We consider the fixed-$T$ test, the adaptive mixture of reduced rank Gaussian martingales, and the plug-in test based on sequential ranks. 
                Based on 2000 simulations from the IFE with $\rho_\varepsilon=\rho_\lambda=0$ and independent standard normal noise. 
                The sample dimensions are $T_0=50,T_\mathcal{B}=25,N=20$.}
                \label{fig:dyn_delta}
            \end{figure}
\newpage
\section{Discussion}
    This paper introduces anytime-valid testing for program evaluation. 
    We construct anytime-valid $p$-values and derive conditions under which these $p$-values have optimal shrinkage rates. 
    In the case of composite alternatives, we propose two strategies: either learning the most promising alternative, or transforming a composite alternative on the outcomes to a simple alternative on the ranks. 
    We show how these methods are often preferred over a standard test that requires a pre-specified number of observations $T$ for a simple intertemporal discounted notion of power.

    In future work, we plan to more extensively study the setting with a dynamic treatment effect. 
    Such settings are challenging in the fixed-$T$ setting because the correct value for $T$ is harder to find, and challenging in the sequential setting as it makes the alternative hard to learn / pre-specify. 
    A fruitful idea may be to design techniques that are tailored to sequentially learning dynamic treatment effects.

    Finally, we note that there are multiple techniques to further increase power based on test martingales. 
    For a given significance level $\alpha$, \cite{fischer2024improving} show that the power of an anytime-valid test can be slightly, but almost surely, improved by avoiding `overshooting'.
    This is done by capping the maximum of our test statistics such that the martingale never exceeds $\frac{1}{\alpha}$, thereby strictly increasing the density of the $e$-value in other regions.
    In addition, if after a certain time $T$, one wishes to terminate the experiment, randomized Markov's inequality can be used to have an almost surely lower critical value at time $T$ \citep{ramdas2023randomized}.
    
\bibliographystyle{apalike} 
\bibliography{references.bib}  

\begin{appendix}

\newpage

\section{Proofs}\label{sec:proofs}
    \subsection{Proof of Theorem \ref{th:rank_test}}\label{sec:prranks}
        We first show that $e_t$ is a sequential $e$-value. 
        For every $t >0$, we have under the null hypothesis \eqref{eq:hyp_rank},
        \begin{align*}
            \mathbb{E}_{S_t,R_t}[e_t\mid\mathcal{G}_{t-1}]
                &= \mathbb{E}_{S_t}\left[\mathbb{E}_{R_t}[e_t\mid\mathcal{G}_{t-1},S_t]\mid \mathcal{G}_{t-1}\right]
                = \mathbb{E}_{S_t}\left[\mathbb{E}_{R_t}\left[\frac{S_t(R_t)}{\tfrac{1}{t}\sum_{i=1}^t S_t(i)} \middle| \mathcal{G}_{t-1},S_t\right]\middle|\mathcal{G}_{t-1}\right] \\
                &= \mathbb{E}_{S_t}\left[\frac{\mathbb{E}_{R_t}\left[S_t(R_t) \mid\mathcal{G}_{t-1}, S_t\right]}{\tfrac{1}{t}\sum_{i=1}^t S_t(i)}\right]
                = \mathbb{E}_{S_t}\left[\frac{\mathbb{E}_{R_t}\left[S_t(R_t) \mid\mathcal{G}_{t-1}\right]}{\tfrac{1}{t}\sum_{i=1}^t S_t(i)}\right]\\
                &= \mathbb{E}_{S_t}\left[\frac{\tfrac{1}{t}\sum_{i=1}^t S_t(i)}{\tfrac{1}{t}\sum_{i=1}^t S_t(i)}\right]
                = 1,
        \end{align*}
        where we condition on $S_t$ in the first equality, use that $S_t\perp R_t|\mathcal{G}_{t-1}$ in the fourth equality, and use that $R_t|\mathcal{G}_{t-1}$ is uniform on $\{1, \dots, t\}$ in the fifth equality.
        The non-negativity of $e_t$ follows from the non-negativity of $S_t$.
        Lastly, the denominator can only be zero if the numerator is, since all the summands in the demoninator are non-negative and the denominator $S_t(R_t)$ is one of the summands. 
        Hence $e_t$ is a sequential $e$-value.

        For the log-optimality result, we first recognize that since $S_t$ is $\mathcal{G}_{t-1}$-measurable, $S_t\perp R_t|\mathcal{G}_{t-1}$, such that $e_t$ is indeed a sequential $e$-value. 
        Next, observe that $\sum_{i=1}^t g_t(i|\mathcal{G}_{t-1}) = 1$, because $g_t(\cdot | \mathcal{G}_{t-1})$ is a density w.r.t. the counting measure on $\{1,\dots,t\}$.
        As a consequence, 
        \begin{align*}
            e_t 
                = \frac{g_t(R_t|\mathcal{G}_{t-1})}{\frac{1}{t}\sum_{i=1}^t g_t(i|\mathcal{G}_{t-1})} 
                = \frac{g_t(R_t|\mathcal{G}_{t-1})}{\frac{1}{t}}.
        \end{align*}
        Finally, note that $1/t$ is the uniform density on $\{1, \dots, t\}$, which is exactly the conditional density of $R_t$ under our null hypothesis \eqref{eq:hyp_rank}.
        As a result, this is the sequential likelihood ratio as defined in \eqref{eq:logopt} and so its running product is the log-optimal martingale.
    \subsection{Proof of Theorem \ref{th:red_rank_test}}\label{sec:prredranks}
        First, we note that if under $H_0$, $\widetilde{S}_t \perp \widetilde{R}_t|\mathcal{G}_{t-1}$, it also holds that $\widetilde{S}_t \perp \widetilde{R}_t|\widetilde{\mathcal{G}}_{t-1}$, as the distribution of $\widetilde{S}_t$ only depends on $\mathcal{G}_{t-1}$ under $H_0^\text{red-rank}$.  
        We now show that $\tilde{e}_t$ is a conditional $e$-value. 
        For every $t \geq 0$, we have under the null hypothesis \eqref{eq:redh0},
        \begin{align*}
            \mathbb{E}_{\widetilde{S}_t,\widetilde{R}_t}[\tilde{e}_t|\widetilde{\mathcal{G}}_{t-1}]
                &= \mathbb{E}_{\widetilde{S}_t}\left[\mathbb{E}_{\widetilde{R}_t}[\tilde{e}_t|\widetilde{\mathcal{G}}_{t-1},\widetilde{S}_t]\mid \widetilde{\mathcal{G}}_{t-1}\right]
                = \mathbb{E}_{\widetilde{S}_t}\left[\mathbb{E}_{\widetilde{R}_t}\left[\frac{\widetilde{S}_t(\widetilde{R}_t)/q_t^{\widetilde{R}_t}}{\sum_{i=1}^{T_0+1} \widetilde{S}_t(i)} \middle| \widetilde{\mathcal{G}}_{t-1},\widetilde{S}_t\right]\middle|\widetilde{\mathcal{G}}_{t-1}\right] \\
                &= \mathbb{E}_{\widetilde{S}_t}\left[\frac{\mathbb{E}_{\widetilde{R}_t}\left[\widetilde{S}_t(\widetilde{R}_t)/q_t^{\widetilde{R}_t}\middle|{\widetilde{\mathcal{G}}_{t-1}},\widetilde{S}_t\right]}{\sum_{i=1}^{T_0+1} \widetilde{S}_t(i)} \middle|\widetilde{\mathcal{G}}_{t-1}\right]
                = \mathbb{E}_{\widetilde{S}_t}\left[\frac{\mathbb{E}_{\widetilde{R}_t}\left[\widetilde{S}_t(\widetilde{R}_t)/q_t^{\widetilde{R}_t}\middle|{\widetilde{\mathcal{G}}_{t-1}}\right]}{\sum_{i=1}^{T_0+1} \widetilde{S}_t(i)} \middle|\widetilde{\mathcal{G}}_{t-1}\right]\\         
                &= \mathbb{E}_{\widetilde{S}_t}\left[\frac{\sum_{i=1}^{T_0+1} \widetilde{S}_t(i)}{\sum_{i=1}^{T_0+1} \widetilde{S}_t(i)}\middle|\widetilde{\mathcal{G}}_{t-1}\right]
                = 1,
        \end{align*}
        where we condition on $\widetilde{S}_t$ in the first equality, and in the third equality we use that $\widetilde{S}_t\perp\widetilde{R}_t|\widetilde{\mathcal{G}}_{t-1}$ and that $q_t^1,\dots,q_t^{T_0+1}$ are $\widetilde{\mathcal{G}}_{t-1}$-measurable.  
        The fourth equality plugs in the distribution of $\widetilde{R}_t|\widetilde{\mathcal{G}}_{t-1}$ under $H_0^{\text{red-rank}}$ given in \eqref{eq:redh0}.
        The non-negativity of $\tilde{e}_t$ follows from the non-negativity of $\widetilde{S}_t$.
        Lastly, the denominator can only be zero if the numerator is, since all the summands in the demoninator are non-negative and the denominator $\widetilde{S}_t(\widetilde{R}_t)$ is one of the summands. Hence $\tilde{e}_t$ is a conditional $e$-value.

        We now turn to the log-optimality claim. 
        Again, since $\widetilde{S}_t$ is $\widetilde{\mathcal{G}}_{t-1}$-measurable, we have that $\widetilde{S}_t \perp \widetilde{R}_t|\widetilde{\mathcal{G}}_{t-1}$ such that $\tilde{e}_t$ is a conditional $e$-value.
        Using that $\tilde{g}_t(i|\widetilde{\mathcal{G}}_{t-1})$ is a density w.r.t. the counting measure on $\{1,\dots,T_0+1\}$, it integrates to one: $\sum_{i=1}^{T_0+1}\tilde{g}_t(i|\widetilde{\mathcal{G}}_{t-1})=1$. Rewriting $\tilde{e}_t$,
        \begin{align*}
            \tilde{e}_t=\frac{\tilde{g}_t(\widetilde{R}_t|\widetilde{\mathcal{G}}_{t-1})/q_t^{\widetilde{R}_t}}{\sum_{i=1}^{T_0+1} \tilde{g}_t(i|\widetilde{\mathcal{G}}_{t-1})}=\frac{\tilde{g}_t(\widetilde{R}_t|\widetilde{\mathcal{G}}_{t-1})}{q_t^{\widetilde{R}_t}}.
        \end{align*}
        Since the numerator is the conditional density of the alternative, and the denominator is the conditional density under $H_0^{\text{red-rank}}$, $\tilde{e}_t$ is a conditional likelihood ratio as defined in \eqref{eq:logopt} and hence its running product is a log-optimal test martingale.
        \qed

    \subsection{Proof of Theorem \ref{th:vovkdom}}\label{app:prvovk}
        When reducing from the sequential ranks to the reduced sequential ranks, we lose the information about the placement of the new observation among the previous post-treatment observations. 
        It turns out that we can actually cleanly decompose the sequential ranks into a combination of the reduced sequential ranks and some auxiliary random variable $N_t$ that captures this lost information. 
        The proof strategy will be to show that under the post-exchangeable alternative $H^\text{post-exch.}$, this auxiliary random variable $N_t$ is pure noise. 
        Next, we use Jensen's inequality to show that this noise is harmful in terms of the expected-log power.
        
        Recall that we denote the test statistics for $e_t$ and $\tilde{e}_t$ by $S_t(R_t)$ and $\widetilde{S}_t(\widetilde{R}_t)$, respectively. 
        As we are interested in the log-power, all expectations are taken with respect to $\mathbb{Q}$.
        The proof of this theorem is divided into three steps. 

        In step 1, we show that the normalization constants in $e_t$ and $\tilde{e}_t$ are redundant for the plug-in test statistics, and therefore we can write $e_t(R_t)=S_t(R_t)$ and $\tilde{e}_t(\widetilde{R}_t)=\widetilde{S}(\widetilde{R}_t)/q_t^{\widetilde{R}_t}$.
        In step 2, we show that conditional on the reduced rank $\widetilde{R}_t$ and filtration $\mathcal{G}_{t-1}$, the smoothed version $V_t$ of the sequential rank $R_t$ is uniformly distributed under $H_1^\text{post-exch}$. 
        This is the key argument of the proof: it implies that learning the distribution of ranks with the same reduced rank is futile and can only add noise.
        Finally, step 3 uses the tower property by conditioning on the reduced rank, and applies Jensen's inequality to relate the log-power of the standard plug-in test, $\mathbb{E}[\log e_t|\mathcal{G}_{t-1}]$, to that of the reduced rank plug-in test, $\mathbb{E}[\log \tilde{e}_t|\mathcal{G}_{t-1}]$.

        \paragraph{Step 1: Redundancy normalization constant}
            The difference between the test statistics and their corresponding $e$-values $e_t,\tilde{e}_t$ is that the $e$-values contain a normalization constant that ensures the expectation under the null equals 1. 
            We demonstrate that this normalization is redundant in this setting, by showing that the conditional expectations of $S_t$ and $\widetilde{S}_t/q_t$ are already equal to 1 under $H_0^\text{rank}$.
            
            First, we show that this holds for $S_t$. Recall that under the null $R_t|\mathcal{G}_{t-1}$ is uniform on $\{1,\dots,t\}$, and $u_t\sim U(0,1)$ is independently drawn from $R_t$.
            It follows that, conditional on $\mathcal{G}_{t-1}$, our smoothed rescaled rank $V_t=\frac{R_t-u_t}{t}$ is uniform on $[0,1]$ under $H_0^\text{rank}$. 
            We use this to calculate the conditional expectation of our test statistic $S_t$,
            \begin{align}
                \mathbb{E}[S_t(R_t)|\mathcal{G}_{t-1}]&= \mathbb{E}\left[\widehat{f}_t(V_t(R_t)) \middle|\mathcal{G}_{t-1}\right]= 
                \mathbb{E}\left[\mathbb{E}\left[\widehat{f}_t(V_t(R_t)) \middle|\mathcal{G}_{t-1},\widehat{f}_t\right]\middle| \mathcal{G}_{t-1}\right]\\
                &=\mathbb{E}\left[\widehat{F}(1)-\widehat{F}(0)\middle| \mathcal{G}_{t-1}\right]=1,
            \end{align}
            where we use the law of iterated expectation based on $\tilde{f}_t$ in the second equality. 
            The third equality uses that under the null, $\tilde{f}_t\perp V_t|\mathcal{G}_{t-1}$ and $V_t|\mathcal{G}_{t-1}\sim U(0,1)$.
            Finally, the fourth equality uses that $\widehat{f}_t$ is a density on $[0,1]$.
    
            Next, we show that the conditional expectation of $\widetilde{S}_t(\widetilde{R}_t)/q_t^{\widetilde{R}_t}$ equals 1. 
            Recall from \eqref{eq:redh0} that under $H_0$, $\widetilde{R}_t|\mathcal{G}_{t-1}\sim \text{Cat}\{q_t^1,\dots,q_t^{T_0+1}\}$. 
            Then,
            \begin{align}
                \mathbb{E}[\widetilde{S}_t(\widetilde{R}_t)/q_t^{\widetilde{R}_t}|\mathcal{G}_{t-1}]&= \mathbb{E}\left[\frac{\widehat{F}\left(\sum_{i=1}^{\widetilde{R}_t}q_t^{i}\right)-\widehat{F}\left(\sum_{i=1}^{\widetilde{R}_t-1}q_t^{i}\right)}{q_t^{\widetilde{R}_t}} \middle|\mathcal{G}_{t-1}\right]\\ 
                &=\mathbb{E}\left[\mathbb{E}\left[\frac{\widehat{F}\left(\sum_{i=1}^{\widetilde{R}_t}q_t^{i}\right)-\widehat{F}\left(\sum_{i=1}^{\widetilde{R}_t-1}q_t^{i}\right)}{q_t^{\widetilde{R}_t}} \middle|\mathcal{G}_{t-1},\widehat{f}_t\right]\middle| \mathcal{G}_{t-1}\right]\\
                &=\mathbb{E}\left[\sum_{j=1}^{T_0+1}q_t^j\frac{\widehat{F}\left(\sum_{i=1}^{j}q_t^{i}\right)-\widehat{F}\left(\sum_{i=1}^{j-1}q_t^{i}\right)}{q_t^{j}} \middle| \mathcal{G}_{t-1}\right]\\
                &=\mathbb{E}\left[\widehat{F}\left(\sum_{i=1}^{T_0+1}q_t^{i}\right)-\widehat{F}\left(0\right) \middle| \mathcal{G}_{t-1}\right]\\
                &=\mathbb{E}\left[\widehat{F}\left(1\right)-\widehat{F}\left(0\right) \middle| \mathcal{G}_{t-1}\right]=1
            \end{align}
            where the second equality uses the tower property, and the third equality uses that $\widehat{f}_t\perp \widetilde{R}_t|\mathcal{G}_{t-1}$ and that $\widetilde{R}_t|\mathcal{G}_{t-1}\sim \text{Cat}\{q_t^1,\dots,q_t^{T_0+1}\}$. 
            The fourth equality uses that that the sum in the numerator telescopes and uses the convention that $\sum_{i=1}^0 g(i)=0$.
            Last, the fifth equality uses that $q_t^1,\dots,q_t^{T_0+1}$ sum to one as they are probabilities of a categorical distribution on $T_0+1$ options.
    
            Hence, we can directly write $e_t(R_t)=S_t(R_t)$ and $\tilde{e}_t(\widetilde{R}_t)=S_t(\widetilde{R}_t)/q_t^{\widetilde{R}_t}$.

        \paragraph{Step 2: Uniformity of $V_t$ given previous ranks and current reduced rank\newline}
            \noindent We now derive the distribution of the rescaled smoothed ranks $V_t$, conditional on $\mathcal{G}_{t-1}$ and $\widetilde{R}_t$, under $H_1^\text{post-exch.}$. 
            In particular, we show that this distribution is uniform on the interval $\left[\sum_{i=1}^{\widetilde{R}_t-1}q_t^{i},\sum_{i=1}^{\widetilde{R}_t}q_t^{i}\right]$.

            To do this, we first derive the conditional distribution of the sequential rank given $\mathcal{G}_{t-1}$ and $\widetilde{R}_t$, and later use this to infer the conditional distribution of $V_t$. 
            We start by decomposing the sequential rank $R_t$ into a function of reduced ranks $\widetilde{R}_{T_0+1},...,\widetilde{R}_{t}$ and some auxiliary discrete random variable $N_t$ that adds the information about the placement of the $t$'th observation among other post-treatment observations with the same reduced rank $\widetilde{R}_t$. 
        
            For a given $\widetilde{R}_t$ and $\mathcal{G}_{t-1}$, we can write $R_t$ as follows
            \begin{align}
                R_t &= \sum_{s=T_0+1}^{t-1}I[\widetilde{R}_s<\widetilde{R}_t] + (\widetilde{R}_{t}-1) +N_t
            \end{align}
            where the first term counts the previous post-treatment observations with lower reduced ranks, the second term corresponds to the pre-treatment outcomes that have a lower rank, and $N_t$ is as described above.
        
            As post-treatment data are exchangeable under $H_1^\text{post-exch.}$, and since there are $n+1$ slots to place something between $n$ observations, $N_t$ has the following conditional distribution, 
            \begin{align*}
                N_t|\mathcal{G}_{t-1},\widetilde{R}_{t}\sim \text{Unif} \left\{1,\dots,1+\sum_{s=T_0+1}^{t-1}I[\widetilde{R}_s=\widetilde{R}_t]\right\}.
            \end{align*}
            
            Using the definition of $q_t^i=\frac{1}{t}\left(1+\sum_{s=1}^{t-1}I[\widetilde{R}_s=i]\right)$, the support of $N_t$ can be rewritten more compactly as $\left\{1,...,tq_t^{\widetilde{R}_t}\right\}$.
    
            Knowing the distribution of $N_t$, we now derive the distribution of $R_t$.
            Rearranging terms in $R_t$, we obtain              
            \begin{align}
                R_t &= \sum_{i=1}^{\widetilde{R}_{t-1}}\sum_{s=T_0+1}^{t-1}I[\widetilde{R}_s=i] + (\widetilde{R}_{t}-1) +N_t\\
                &= \sum_{i=1}^{\widetilde{R}_{t-1}}\left(1+\sum_{s=T_0+1}^{t-1}I[\widetilde{R}_s=i]\right) +N_t\\
                &= t\sum_{i=1}^{\widetilde{R}_{t-1}}q_t^i+N_t,
            \end{align}
            where the third equality again uses the definition of $q_t^i$. 
            It follows that conditional on $\mathcal{G}_{t-1}$ and $\widetilde{R}_{t-1}$, the sequential rank $R_t$ is uniformly distributed on $\left\{t\sum_{i=1}^{\widetilde{R}_{t-1}}q_t^i+1,\dots,t\sum_{i=1}^{\widetilde{R}_{t}}q_t^i\right\}$. 
            Since $V_t = \frac{R_t-u_t}{t}$, and $u_t\sim U(0,1)$ is drawn independently from $R_t$, we conclude that 
            \begin{align*}
                V_t|\mathcal{G}_{t-1},\widetilde{R}_{t-1}\sim U\left(\sum_{i=1}^{\widetilde{R}_t-1}q_t^{i},\sum_{i=1}^{\widetilde{R}_t}q_t^{i}\right).
            \end{align*}

        \paragraph{Step 3: Finalizing the proof}
            We now use the findings of step 1 and 2 to conclude the proof. 
            Recall that we need to prove the inequality in $\eqref{eq:vovkineq}$. 
            Starting from the LHS,
            \begin{align}
                 \mathbb{E}\left[\log e_t|\mathcal{G}_{t-1}\right]&=\mathbb{E}\left[\log \widehat{f}_t\left(V_t\right)\middle| \mathcal{G}_{t-1}\right]\\ &= 
                 \mathbb{E}\left[\mathbb{E}\left[\log \widehat{f}_t\left(V_t\right)\middle|\mathcal{G}_{t-1},\widetilde{R}_t\right]\middle|\mathcal{G}_{t-1}\right]\\
                 &\leq \mathbb{E}\left[\log\mathbb{E}\left[ \widehat{f}_t\left(V_t\right)\middle|\mathcal{G}_{t-1},\widetilde{R}_t\right]\middle|\mathcal{G}_{t-1}\right] \\
                 &=\mathbb{E}\left[\log\left(\int_{\sum_{i=1}^{\widetilde{R}_t-1}q_t^{i}}^{\sum_{i=1}^{\widetilde{R}_t}q_t^{i}}\frac{\widehat{f}_t\left(v\right)dv}{q_t^{\widetilde{R}_t}}\right)\middle|\mathcal{G}_{t-1}\right]\\
                 &= \mathbb{E}\left[\log\left(\frac{\widehat{F}\left({\sum_{i=1}^{\widetilde{R}_t}q_t^{i}}\right)-\widehat{F}\left({\sum_{i=1}^{\widetilde{R}_t-1}q_t^{i}}\right)}{q_t^{\widetilde{R}_t}}\right)\middle|\mathcal{G}_{t-1}\right],\\
                 &=\mathbb{E} \left[\log\tilde{e}_t|{\mathcal{G}}_{t-1}\right]
            \end{align}
            where the first equality uses the result from step 1, the second equality applies the tower property, and the first inequality follows from Jensen's inequality. 
            The third equality uses the result from step 2, that $V_t |\mathcal{G}_{t-1},\widetilde{R}_{t}$ is uniformly distributed on $\left(\sum_{i=1}^{\widetilde{R}_t-1}q_t^{i},\sum_{i=1}^{\widetilde{R}_t}q_t^{i}\right)$.
            The third equality also uses the fact that $\tilde{f}_t$ is independent of $\widetilde{R}_t$, given $\mathcal{G}_{t-1}$.  
            The fourth equality is based on the definition of the CDF $\widehat{F}(\cdot)$. 
            Finally, the fifth equality uses the result from step 1.\qed

\newpage

\section{Minimax Regret Property for Test Martingales}\label{app:minimax}
    The following proposition shows that the adaptive mixture over $k$ $e$-values that attains minimax `regret' is a simple average over their corresponding martingales. Regret is defined here as the difference between the growth rate of the mixture martingale and that of the fastest growing candidate martingale.
    Note that this mixture still constitutes data-dependent mixing over the $e$-values, as this mixing is done proportionally to the size of their martingales at time $t-1$.
    The regret bound also implies that the average growth of the mixture martingale converges to that of the fastest growing test martingale.
    \begin{proposition}\label{th:mixing}
        Suppose $(W_t^1)_{t\ge 0},\dots,(W_t^k)_{t\ge 0}$ are test martingales for $H_0$ adapted to the common filtration $\mathcal{F}=(\mathcal{F}_t)_{t\ge 0}$.  
        For $j=1,\dots,k$ and $t\geq1$, set $e_t^j:=W_t^j/W_{t-1}^j$ and assume that these $e$-values have support on $[0,\infty)$. Define  
        \begin{equation*}
            \widetilde{W}_t:=\frac{1}{k}\sum_{j=1}^k W_t^j,\qquad t\ge 0 .
        \end{equation*}
        Then, $(\widetilde{W}_t)_{t\ge 0}$ is a test martingale  that enjoys the following regret bound for every $t\ge 1$,
        \begin{equation}\label{eq:regret}
            \mathcal{R}_t(\widetilde{W}_t) := \max_{1\le j\le k}\left\{\log W_t^j\right\}-\log\widetilde{W}_t\le\log k .
        \end{equation}
        Also, $\log k$ is the tightest minimax regret bound among all running products that use $\mathcal{F}_{t-1}$-measurable convex combinations of $e_t^1,\dots,e_t^k$.
    \end{proposition}

    \begin{proof}
        Non-negativity and unit-start of $(\widetilde{W_t})_{t\geq0}$ are directly inherited from the candidate martingales. 
        The martingale property is also retained; for any $t\geq 1$,
        \begin{align*}
           \sup_{\mathbb{P}\in H}\mathbb{E}^\mathbb{P}\left[\widetilde{W}_t\mid\mathcal{F}_{t-1}\right]
           &=    \sup_{\mathbb{P}\in H}\mathbb{E}^\mathbb{P}\left[\frac1k\sum_{i=1}^k W_t^j
           \middle|\mathcal{F}_{t-1}\right]
           \\
           &\leq\frac1k\sum_{i=1}^k\sup_{\mathbb{P}\in H}\mathbb{E}^\mathbb{P}\!\bigl[W_t^i\mid\mathcal{F}_{t-1}\bigr]
           \le\frac1k\sum_{i=1}^k W_{t-1}^i
           =\widetilde{W}_{t-1},
        \end{align*}
        hence $\{\widetilde{W}_t\}_{t\geq0}$ is a test martingale.
        
        The regret bound follows directly from the definition of $\{\widetilde{W}_t\}_{t\geq0}$,
        \begin{equation*}
            \max_{1 \le j\le k}\left\{\log W_t^j\right\}-\log\widetilde{W}_t = 
            \log \left(\frac{\max_{1 \le j\le k}\left\{ W_t^j\right\}}{\frac{1}{k} \sum_{i=1}^kW_t^i}\right)
            \leq  \log \left(\frac{\max_{1 \le j\le k}\left\{ W_t^j\right\}}{\frac{1}{k} \max_{1 \le j\le k}\left\{ W_t^j\right\}}\right) = \log k.
        \end{equation*}
        To show that this is the tightest possible regret bound, we show that any mixture strategy can incur a regret of $\log k$ at any $t$.
        For all $t\geq1$, denote with $\bm{\lambda}_t=(\lambda_t^1,\dots,\lambda^k_t)$ a vector of $\mathcal{F}_{t-1}$-measurable convex weights.
        Let $\widehat{W}_t$ denote the martingale corresponding to this mixture strategy,
        \begin{equation*}
            \widehat{W}_t := \prod_{s=1}^t \sum_{i=1}^k\lambda_{s}^i e_s^i,\quad \text{for all $t\geq1$.}
        \end{equation*}
        
        Now choose $j_*:=\arg\min_{1\leq j\leq k}\lambda_1^j$; then $\lambda_1^{j_*}\le1/k$.
        Consider the following realizations of the candidate $e$-values,
        \begin{equation*}
            e_t^i = \begin{cases}
                        0, \quad \text{if $t=1$ and $i\neq j_*$} \\
                        1,\quad \text{else}.
                    \end{cases}
        \end{equation*}
        Then, for all $t\geq 1$, $\widehat{W}_t = \lambda_1^{j_*}$ and $\max_{1\leq j\leq k }W_{t}^i=1$. 
        So for any weighting strategy, there exists a realization of the candidate $e$-values for which,
        \begin{equation*}
            \max_{1\le j\le k}\left\{\log W_t^j\right\}-\log\widehat{W}_t= 0-\log \lambda_1^{j_*}\ge\log k,
        \end{equation*}
        and hence $\log k$ is the tightest possible regret bound.
        \end{proof}

\newpage

\newpage

\section{Efficient approximation of Gaussian reduced rank test statistic}\label{app:effgauss}
    Consider the reduced rank Gaussian setting as described in Section \ref{sec:gauss}. 
    First, we condition on the pre-treatment outcomes $Y^\text{pre}=Y_1,\dots,Y_{T_0}$. 
    Then, the conditional probability of observing reduced rank $\widetilde{R}_t=\tilde{r}$ is given by 
    \begin{align}
        g_t(\tilde{r}|Y^\text{pre}) = \Phi\left(Y^\text{pre}_{[\tilde{r}]}-\frac{\mu_\text{post}}{\sigma}\right)-\Phi\left(Y^\text{pre}_{[\tilde{r}-1]}-\frac{\mu_\text{post}}{\sigma}\right),
    \end{align}
    where $\Phi$ is the standard normal CDF and with slight abuse of notation $Y^{\text{prev}}_{[0]}=-\infty$ and $Y^{\text{prev}}_{[T_0+1]}=\infty$.
        
    Using the law of conditional expectation, the conditional probabilities that $\widetilde{R}_t=\tilde{r}_t$ given $\widetilde{R}_{1}=\tilde{r}_1,\dots,\widetilde{R}_{t-1}=\tilde{r}_{t-1}$ are proportional to the following expectation,
        
    \begin{align*}
        \widetilde{f}_t(\tilde{r}_t)\propto \mathbb{E}^{Y^\text{pre}}\left[\prod_{j=1}^{T_0+1}g_t(j|Y^\text{pre})^{\#\{i \in \{T_0+1,...,t\}|\tilde{r}_i=j\}}\right].
    \end{align*}

    Below, we provide an algorithm for efficiently finding this expectation.
    \newpage
    \begin{algorithm}[ht!]
        \caption{Efficient computation of optimal Gaussian test statistic}\label{alg:reduced_rank}
        \KwIn{
        \begin{itemize}
            \item Time index: $t > T_0$
            \item Previously observed reduced ranks: $\tilde{r}_{s}$, for $s < t$
            \item Effect size: $\frac{\mu_\text{post} - \mu_\text{pre}}{\sigma}$
            \item Number of pre-treatment observations: $T_0$
            \item Sample size: $T$
        \end{itemize}
        }
        \KwOut{Optimal test statistic $\widetilde{S}_t \propto \tilde{f}_t$}

        \BlankLine
        \For{$m = 1, \dots, M$}{
        Draw $X_1^m,\dots,X_{T_0}^m \overset{\text{i.i.d.}}{\sim} \mathcal{N}(0,1)$.
        }

        \BlankLine

        \For{$\widetilde{r} = 1, \dots, T_0 + 1, m = 1, \dots, M$}{
         Compute the conditional probability $g_t(\widetilde{r} \mid X_1^m, \dots X_{T_0}^m)$ using:
        \begin{align*}
             g_t(\widetilde{r} \mid X_1^m,\dots X_{T_0}^m) &= \Phi\left( X^m_{[\widetilde{r}]} - \frac{\mu_\text{post}-\mu_\text{pre}}{\sigma} \right) \\
             &-  \Phi\left( X^m_{[\widetilde{r}-1]} - \frac{\mu_\text{post}
             -\mu_\text{pre}}{\sigma} \right)
         \end{align*}
         where $X^m_{[0]} = -\infty$ and $X^m_{[T_0+1]} = \infty$.
        }
        \BlankLine

        \For{$\widetilde{r}_t = 1, \dots, T_0 + 1$}{
       Compute the test statistic for $\widetilde{R}_t = \widetilde{r}_t$ given the previous reduced ranks $\widetilde{r}_1, \dots, \widetilde{r}_{t-1}$:
        \begin{align*}
            \widetilde{S}_t(\widetilde{r}_t) = \frac{1}{M} \sum_{i=1}^M \left[ \prod_{j=1}^{T_0+1} g_t(j \mid X_1^m,\dots X_{T_0}^m)^{\#\{i \in \{T_0+1, \dots, t\} \mid \tilde{r}_i = j \}} \right]
        \end{align*}
        }
    \end{algorithm}

\newpage

\section{Additional simulation results}
    \subsection{Size distortions without block-structure}\label{app:blocks}
        We study the size properties under the same simulation setting as in Table \ref{tab:b3}, but now with a block size of $B=1$. 
        
        From Table $\ref{tab:b1}$, we again see that all tests maintain size control across all sample dimensions, when $\rho_\lambda=\rho_\varepsilon=0$. 
        However, in the non-exchangeable settings, the size distortions are more pronounced. 
        We see that there is substantially more size distortion compared to when $B=3$. 
        For example, the size of the anytime-valid Gaussian test in the `least exchangeable' setting where $\rho_\lambda=0.75, \rho_\varepsilon=0.5$, $T_0=20$, and $N=20$, is now 20\%, as opposed to $9$\% for $B=3$. 
        
        It seems that when $B=1$, the test mainly suffers from a high correlation in the idiosyncratic terms ($\rho_\varepsilon)$. 
        When the pre-treatment sample $T_0$ is large and $\rho_\varepsilon$ is low, the methods still control size reasonably well. 
        The relative performance of the three tests is similar to when $B=3$.
        
        \begin{table}[h!]
\centering
\caption{Rejection rates under $H_0$ $(\alpha=0.05)$ with $B=1$} \label{tab:b1}
\small
\begin{tabular}{@{}cc|ccccccccccc@{}}
\toprule
\midrule
&&\multicolumn{9}{c}{$\bm{\rho_\varepsilon=0}$} \\
\midrule
&& \multicolumn{3}{c}{Fixed-$T$} & \multicolumn{3}{c}{AV Gaussian} & \multicolumn{3}{c}{AV Plug-in} \\ 
\cmidrule(lr){3-5} \cmidrule(lr){6-8} \cmidrule(lr){9-11}
 $\rho_\lambda$ & $N$ & $T_0=20$ & $T_0=50$ & $T_0=100$ & $T_0=20$ & $T_0=50$ & $T_0=100$ & $T_0=20$ & $T_0=50$ & $T_0=100$  \\
\midrule
\multirow{2}{*}{0} &20 & 0.06 & 0.06 & 0.05 & 0.03 & 0.03 & 0.02 & 0.02 & 0.03 & 0.02 \\ 
 &50 & 0.05 & 0.05 & 0.04 & 0.03 & 0.03 & 0.02 & 0.02 & 0.03 & 0.02 \\ 
\midrule
\multirow{2}{*}{0.25} &20 & 0.06 & 0.06 & 0.06 & 0.05 & 0.04 & 0.03 & 0.03 & 0.03 & 0.03 \\ 
 &50 & 0.07 & 0.07 & 0.07 & 0.05 & 0.04 & 0.03 & 0.03 & 0.03 & 0.03 \\ 
\midrule
\multirow{2}{*}{0.5} &20 & 0.09 & 0.09 & 0.08 & 0.07 & 0.06 & 0.05 & 0.06 & 0.06 & 0.04 \\ 
 &50 & 0.09 & 0.07 & 0.08 & 0.07 & 0.04 & 0.03 & 0.06 & 0.03 & 0.05 \\ 
\midrule
\multirow{2}{*}{0.75} &20 & 0.13 & 0.12 & 0.12 & 0.12 & 0.10 & 0.08 & 0.14 & 0.13 & 0.11 \\ 
 &50 & 0.13 & 0.11 & 0.09 & 0.12 & 0.07 & 0.05 & 0.12 & 0.08 & 0.07 \\ 
\midrule
$\rho_\lambda$ & $N$ & \multicolumn{9}{c}{$\bm{\rho_\varepsilon=0.25}$} \\ 
\midrule
\multirow{2}{*}{0} &20 & 0.09 & 0.10 & 0.09 & 0.06 & 0.07 & 0.05 & 0.04 & 0.06 & 0.05 \\ 
 &50 & 0.08 & 0.09 & 0.09 & 0.06 & 0.06 & 0.05 & 0.04 & 0.05 & 0.06 \\ 
\midrule
\multirow{2}{*}{0.25} &20 & 0.10 & 0.10 & 0.11 & 0.08 & 0.07 & 0.07 & 0.07 & 0.07 & 0.06 \\ 
 &50 & 0.10 & 0.10 & 0.12 & 0.08 & 0.06 & 0.05 & 0.07 & 0.06 & 0.06 \\ 
\midrule
\multirow{2}{*}{0.5} &20 & 0.13 & 0.13 & 0.11 & 0.12 & 0.10 & 0.07 & 0.11 & 0.10 & 0.08 \\ 
 &50 & 0.11 & 0.11 & 0.11 & 0.11 & 0.07 & 0.06 & 0.10 & 0.08 & 0.08 \\ 
\midrule
\multirow{2}{*}{0.75} &20 & 0.17 & 0.15 & 0.15 & 0.17 & 0.12 & 0.11 & 0.18 & 0.16 & 0.16 \\ 
 &50 & 0.16 & 0.14 & 0.14 & 0.15 & 0.11 & 0.10 & 0.17 & 0.15 & 0.12 \\ 
\midrule
$\rho_\lambda$ & $N$ & \multicolumn{9}{c}{$\bm{\rho_\varepsilon=0.5}$} \\ 
\midrule
\multirow{2}{*}{0} &20 & 0.13 & 0.15 & 0.16 & 0.11 & 0.13 & 0.11 & 0.10 & 0.13 & 0.14 \\ 
 &50 & 0.14 & 0.14 & 0.16 & 0.12 & 0.10 & 0.09 & 0.11 & 0.12 & 0.14 \\ 
\midrule
\multirow{2}{*}{0.25} &20 & 0.15 & 0.15 & 0.14 & 0.13 & 0.13 & 0.11 & 0.13 & 0.15 & 0.15 \\ 
 &50 & 0.16 & 0.16 & 0.15 & 0.15 & 0.13 & 0.10 & 0.14 & 0.15 & 0.14 \\ 
\midrule
\multirow{2}{*}{0.5} &20 & 0.17 & 0.16 & 0.16 & 0.15 & 0.14 & 0.13 & 0.17 & 0.19 & 0.18 \\ 
 &50 & 0.16 & 0.16 & 0.17 & 0.15 & 0.13 & 0.12 & 0.17 & 0.18 & 0.17 \\ 
\midrule
\multirow{2}{*}{0.75} &20 & 0.20 & 0.21 & 0.19 & 0.20 & 0.19 & 0.15 & 0.25 & 0.26 & 0.24 \\ 
 &50 & 0.19 & 0.19 & 0.18 & 0.19 & 0.18 & 0.13 & 0.25 & 0.24 & 0.21 \\ 
\midrule
\bottomrule
\end{tabular}
\begin{tablenotes} 
    \item\textit{Notes:} 2000 repetitions of simulation design from IFE model as described in (\ref{eq:scm1},\ref{eq:scm2}). Fixed-$T$: test by \cite{abadie2021synthetic}, AV Gaussian: Anytime-valid reduced rank test with Gaussian alternative, AV Plug-in: Anytime-valid reduced rank test with plug-in alternative. The fixed-$T$ test is performed after $12$ post-treatment observations. $T_\mathcal{B}=T_0/2$ and $T-T_0=30$.
\end{tablenotes}
\end{table}

        \newpage

        \subsection{Similar rejection rates for SCM and DiD}\label{app:similar}
        
        \begin{figure}[ht!]
            \centering
            \begin{subfigure}[b]{0.45\linewidth}
                \centering
                \includegraphics[width=\linewidth]{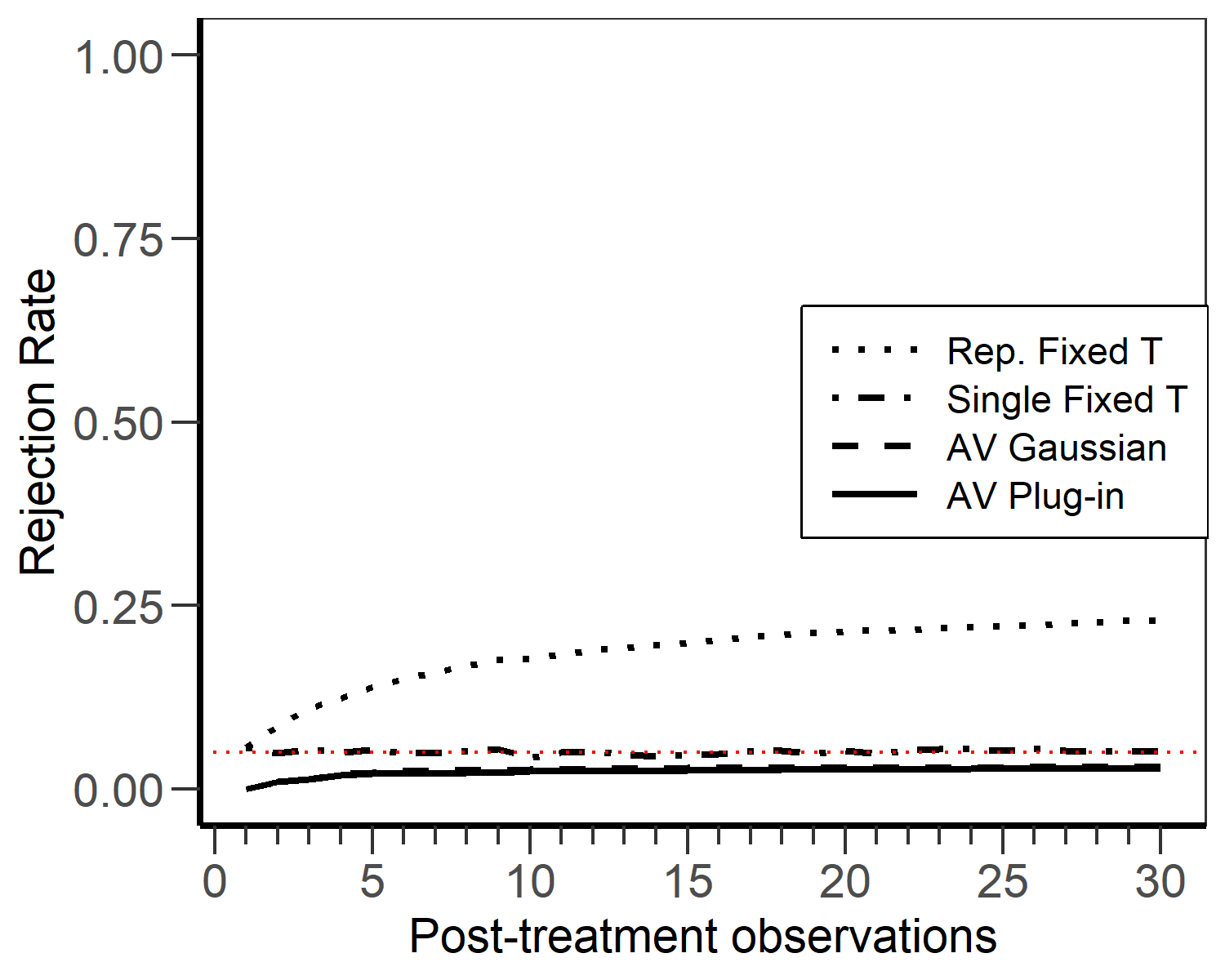}
                \caption{ $H_0: \tau_t=0$}
                \label{fig:apsim1}
            \end{subfigure}%
            \hfill
            \begin{subfigure}[b]{0.45\linewidth}
                \centering
                \includegraphics[width=\linewidth]{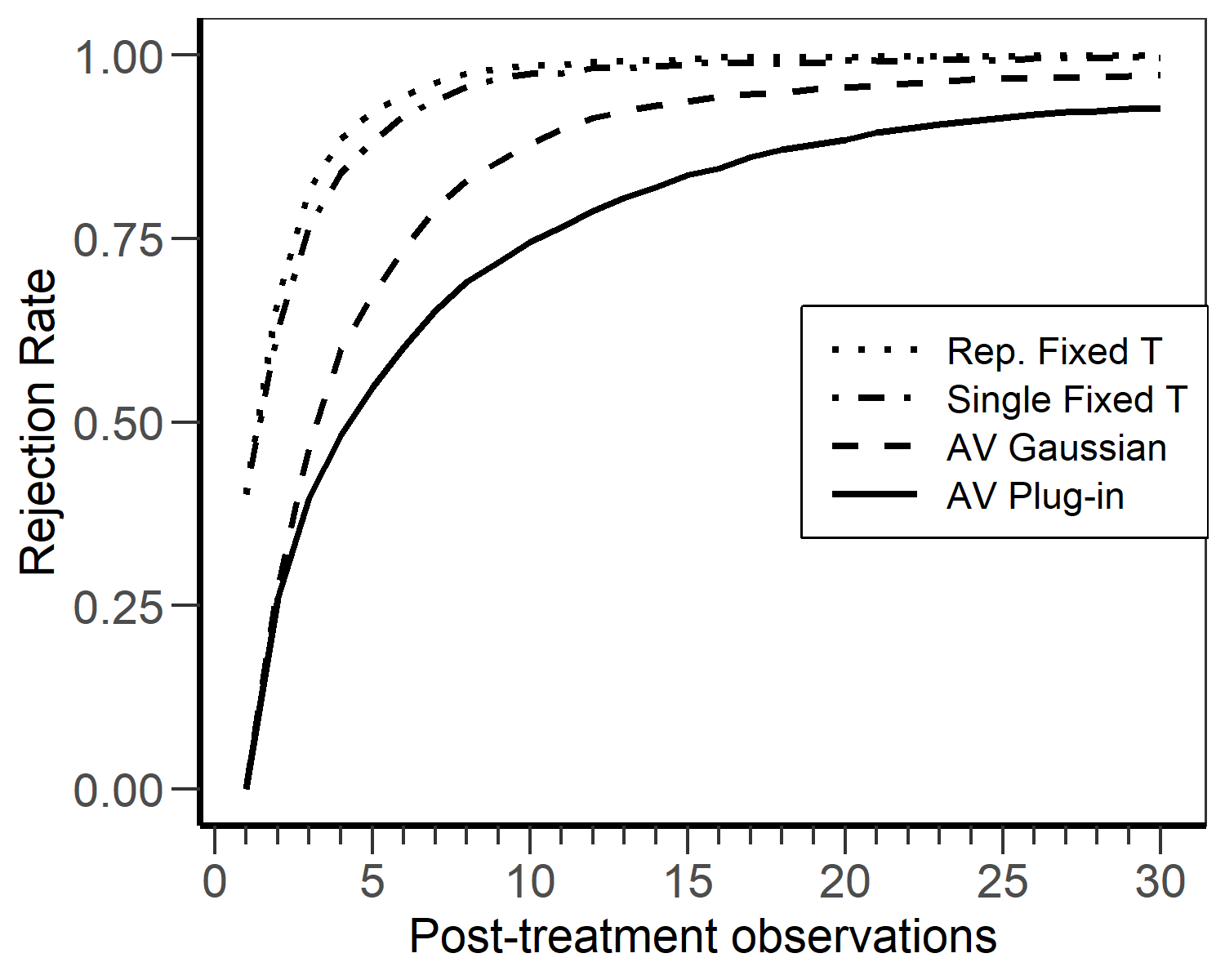}
                \caption{$H_1: \tau_t=2$}
                \label{fig:apsim2}
            \end{subfigure}
            \caption{Rejection rates based on the SCM estimator under $H_0$ and $H_1$ for different testing procedures ($\alpha=0.05$) over time.
                Repeated fixed-$T$: rejection as soon as $p_t$ of fixed-$T$ dips below $\alpha$. 
                Single fixed-$T$: rejection if fixed-$T$ test rejects at time $T$. 
                AV Gaussian: anytime-valid reduced rank test with Gaussian alternative.
                AV Plug-in: anytime-valid reduced rank test with plug-in alternative. 
                Results are based on 2000 simulations from the IFE model with parallel trends and independent standard normal noise. 
                Sample sizes: $T_0=50$, $T_\mathcal{B}=25$, $N=20$. }
            \label{fig:app_comp}
        \end{figure}
        
        \begin{figure}[ht!]
            \centering
            \begin{subfigure}[b]{0.49\linewidth}
                \centering
                \includegraphics[width=\linewidth]{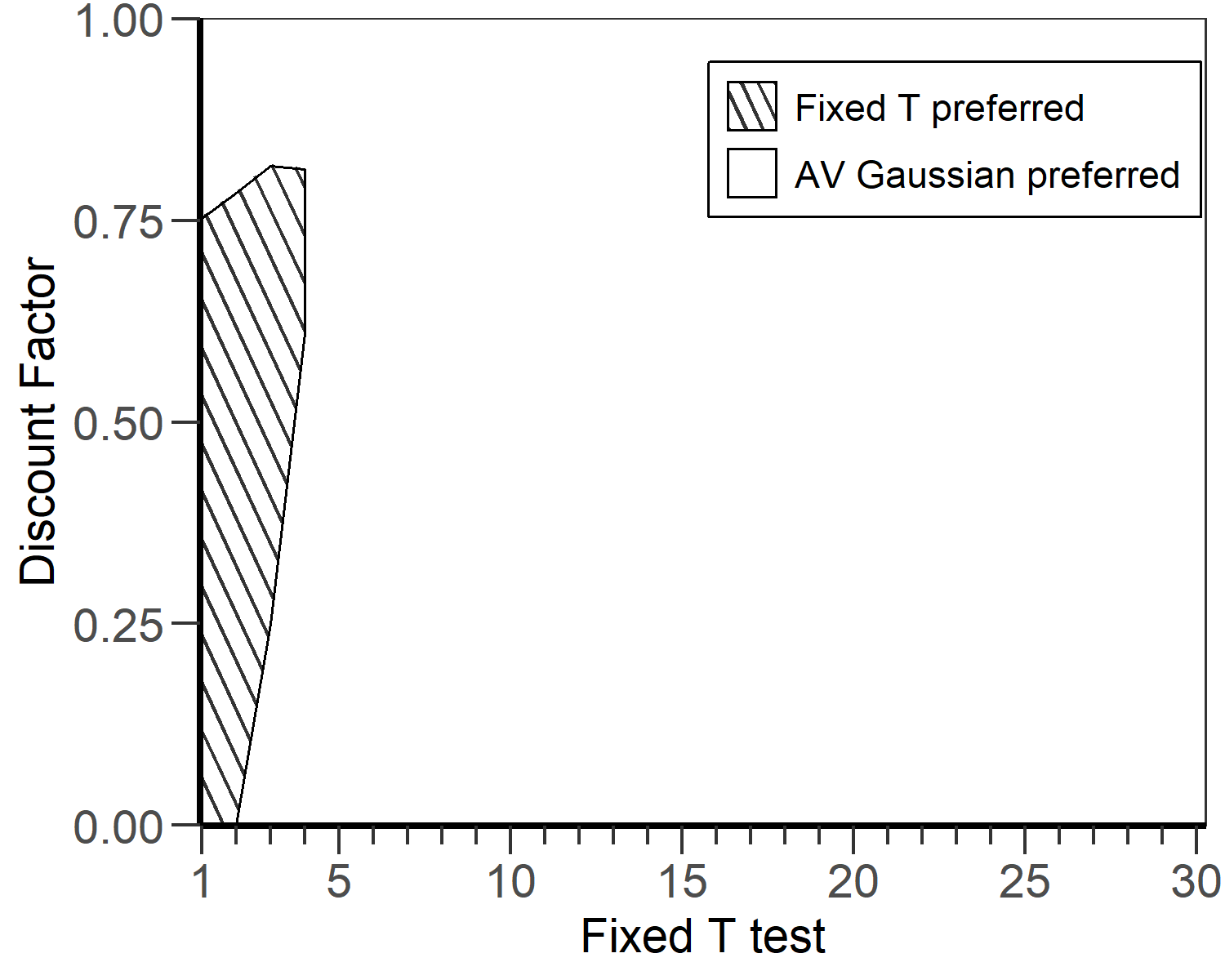}
            \caption{Anytime-valid test with Gaussian alternative vs. multiple fixed-$T$ tests}
                \label{fig:apdelta1}
            \end{subfigure}%
            \hfill
            \begin{subfigure}[b]{0.49\linewidth}
                \centering
                \includegraphics[width=\linewidth]{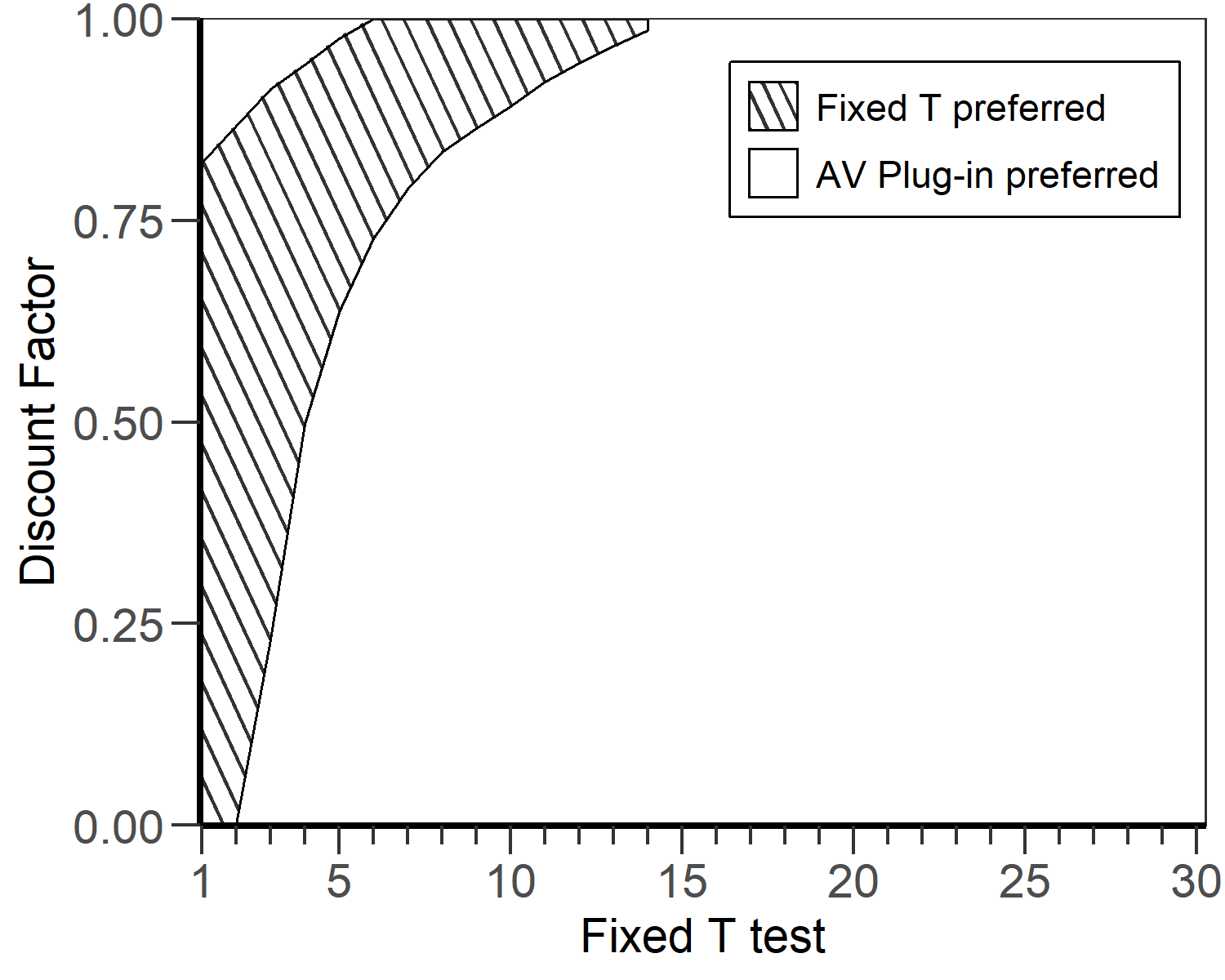}
            \caption{Anytime-valid test with plug-in alternative vs. multiple fixed-$T$ tests}
                \label{fig:apdelta2}
            \end{subfigure}
            \caption{Regions of discount factors $\delta$ in which the fixed-$T$ test at different levels of $T$ attains higher discounted utility than the anytime-valid tests. 
            The simulation settings are the same as in Figure \ref{fig:apsim2}.}
            \label{fig:app_delta}
        \end{figure}

\newpage 

 \section{Fixed-\textit{T} test} \label{app:fixedT}
    In this section, we present a fixed-$T$ test that can be used for inference within the IFE. 
    Specifically, we focus on the permutation test proposed by \cite{abadie2021synthetic}, which, like our approach, is based on split conformal inference. 
    This stands in contrast to the full conformal inference approach adopted by \cite{chernozhukov2021exact}.
    
    Let $
    \widehat{\boldsymbol{\tau}} = \left\{\widehat{\tau}_t \right\}_{t \in \mathcal{B} \cup \{T_0+1, \dots, T\}}$
    represent a sequence of treatment effect estimators, which are assumed to be exchangeably distributed under the null hypothesis, $H_0$. 
    Define $\Pi$ as the set of all possible length-$(T - T_0)$-combinations of the elements in $\mathcal{B} \cup \{T_0+1, \dots, T\}$. For each $\pi \in \Pi$, let $\pi(i)$ indicate the $i$th smallest element in $\pi$. 
    We then express 
    $\widehat{\bm{e}}_\pi = \left( \widehat{\tau}_{\pi(1)}, \dots, \widehat{\tau}_{\pi(T-T_0)} \right)$ and
    $
    \widehat{\bm{e}} = \left( \widehat{\tau}_{T_0+1}, \dots, \widehat{\tau}_{T} \right).$
    
    The test statistic used by \cite{abadie2003economic} is given by
    \[
    \bar{S}_t(\bm{e}) = \sum_{t=1}^{T-T_0} |e_t|,
    \]
    though they note that more general $L_p$ norms can be employed. 
    For one-sided tests, the absolute value operator can be put outside of the summation. 
    Now, the permutation $p$-value constructed as
    \[
    \widehat{p} = \frac{1}{|\Pi|} \sum_{\pi \in \Pi} 1\{\bar{S}(\widehat{e}_\pi) \geq \bar{S}(\widehat{e})\},
    \]
    is valid for $H_0$.
 \end{appendix}

\end{document}